\documentclass[12pt]{article}
\usepackage[utf8]{inputenc}
\usepackage[T1]{fontenc}
\usepackage{lmodern}
\usepackage[top=2.5cm, bottom=2.5cm, left=2.5cm, right=2.5cm]{geometry}
\usepackage[english]{babel}
\usepackage{pdfpages}
\usepackage{graphicx}
\usepackage{caption}
\usepackage{subcaption}
\usepackage{hyperref}
\usepackage{amsmath}
\usepackage{amsthm}
\usepackage{amsfonts}
\usepackage{float}
\usepackage{array}
\usepackage{url}
\usepackage{lineno}
\usepackage[protrusion=true,expansion=true]{microtype}

\newtheorem{theorem}{Theorem}[section]

\begin{document}

\begin{titlepage}
    \title{Complexity of Jelly-No and Hanano games with various constraints\footnote{A preliminary version of this article was presented in the Indonesia-Japan Conference on Discrete and Computational Geometry, Graphs, and Games (IJDCG3) 2023.} \footnote{This research was supported by the ANR project P-GASE (ANR-21-CE48-0001-01)}}
    \author{Owen Crabtree, Sorbonne Universit\'e, CNRS, LIP6, F-75005 Paris, France \and Valia Mitsou, IRIF, Universit\'e Paris Cit\'e, UMR 8243, Paris, France}
    \date{April 13, 2026}
    \maketitle
\end{titlepage}

\baselineskip=16pt


\begin{abstract}
    
    \textbf This work shows new results on the complexity of games Jelly-No and Hanano with various constraints on the size of the board and number of colours. 

    Hanano and Jelly-No are one-player, 2D side-view puzzle games with a dynamic board consisting of coloured, movable blocks disposed on platforms. These blocks can be moved by the player and are subject to gravity. Both games, created by Qrostar and available online, somehow vary in their gameplay, but the goal is always to move the coloured blocks in order to reach a specific configuration and make them interact with each other or with other elements of the game. In Jelly-No the goal is to merge all blocks of the same colour, which happens when they make contact. In Hanano the goal is to make all the coloured blocks bloom by making contact with flowers that have the same colour. 

    Jelly-No was proven by Chao Yang to be NP-complete under the restriction that all movable blocks have the same colour and NP-hard for more colours. Hanano was proven by Michael C. Chavrimootoo to be PSPACE-complete under the restriction that all movable blocks have the same colour. However, the question of PSPACE-completeness for Jelly-No with more than one colours was left open.
    
    In this paper, we settle this question, proving that Jelly-No is PSPACE-complete with an unbounded number of colours. We further show that, if we allow black jellies (that is, jellies that cannot and do not need to merge), the game is PSPACE-complete even for one colour. We further show that one-colour Jelly-No and Hanano remain NP-hard even if the width or the height of the board are constants.

\end{abstract}

\section{Introduction}

    \subsection{One-player puzzle games and computational complexity} \label{intro1}

        Hanano and Jelly-No are both one-player puzzle games with a dynamic board consisting of coloured, movable blocks that can be moved to the right or left by the player and are subject to gravity. These games somehow vary in their gameplay, but the goal is always to move the coloured blocks in order to reach a specific configuration and make them interact with each other or with other elements of the game.     
        One-player puzzle games such as Hanano and Jelly-No can usually be linked to the following decision question: is a given level of this game solvable? If so, which succession of configurations can lead to solving the puzzle, i.e. which sequence of moves does the player need to do to reach a winning configuration? 
        
        We study these games either in the general case with no bound on the size of the board or the number of colours, or under some specific constraints. We consider the size of our input to be polynomial relatively to the dimensions of the board.
    
        Given a level of such one-player puzzle games, the memory needed to describe a configuration and check that one configuration is a direct successor of another is polynomial relatively to the size of the board: therefore these games are in PSPACE (a rigorous proof of this fact for the game Jelly-No will be given later on). To study the complexity of such games, it is interesting to know whether there exists for every solvable level a polynomial bound on the minimum number of moves leading to a winning configuration. If such a bound exists, then the decision problem corresponding to the game is in NP: it might not be possible to provide a solution in polynomial time, but for all solvable levels, there exists a succession of moves leading to a winning configuration that can be verified in polynomial time.
    
        However, in some one-player games, the number of possible configurations is exponential. This is usually the case when the moves are reversible. In this case, the computational complexity of a game can jump from NP to PSPACE: a level can be solved using a polynomial amount of space, but solutions might involve an exponential number of moves and therefore cannot even be checked in polynomial time.


    \subsection{The Jelly-No and Hanano puzzles} \label{gameplay}
    
        Jelly-No and Hanano are 2D side-view puzzle games represented by a grid of size $n\times{m}$ consisting of coloured blocks of different sizes and types, either disposed on fixed platforms or attached to walls. Those blocks can move to the left or right direction one step at a time. Besides, they are subject to gravity: in the absence of support from a block or platform under them, they fall until they reach one. Both games were developed by videogame creator Qrostar \cite{Qrostar} and are available online.
        
        In Jelly-No (see Figure \ref{fig:intro_jelly}), coloured blocks are called jellies. When two jellies of the same colour come into contact, they merge into a jelly whose dimensions are the sum of those of the original jellies (see Figure \ref{jelfusion}). Those jellies can sometimes be attached to a wall, and therefore impossible to move. 
       Some levels of the game include black jellies which can move and are subject to gravity, but cannot and do not need to merge together, as visible on Figure \ref{jellyexample}. Note that those black jellies are not necessarily rectangular: see levels 16, 17 or 36 \cite{Qrostar}.
        \begin{figure}[H]
            \begin{subfigure}[t]{0.6\textwidth}
            \centering
             \includegraphics[width=0.7\textwidth]{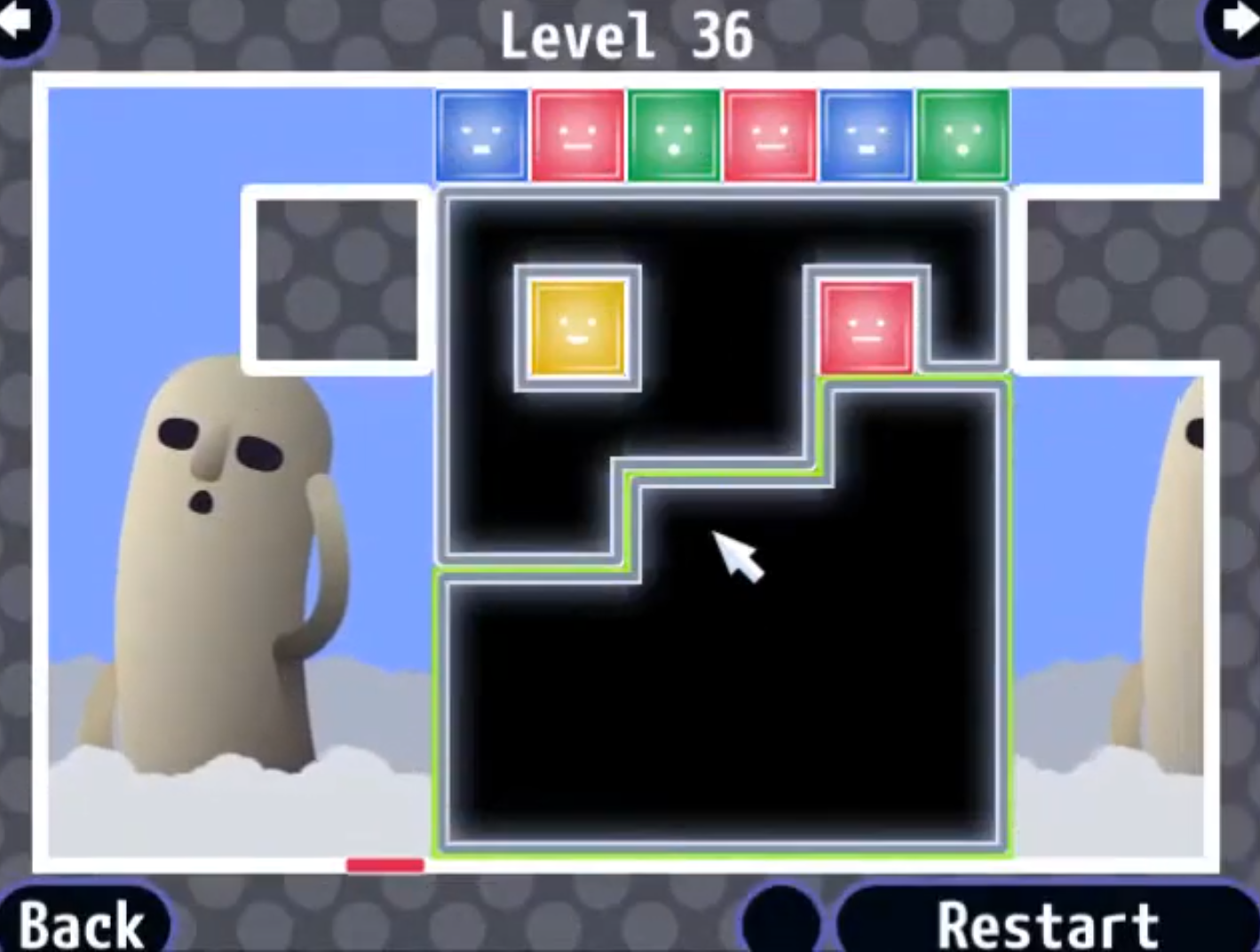}
             \subcaption[]{Level 36 of Jelly-No, with black jellies of various shapes}
            \label{jellyexample}
            \end{subfigure}
            \begin{subfigure}[t]{0.4\textwidth}
                \centering
                 \includegraphics[width=0.9\textwidth]{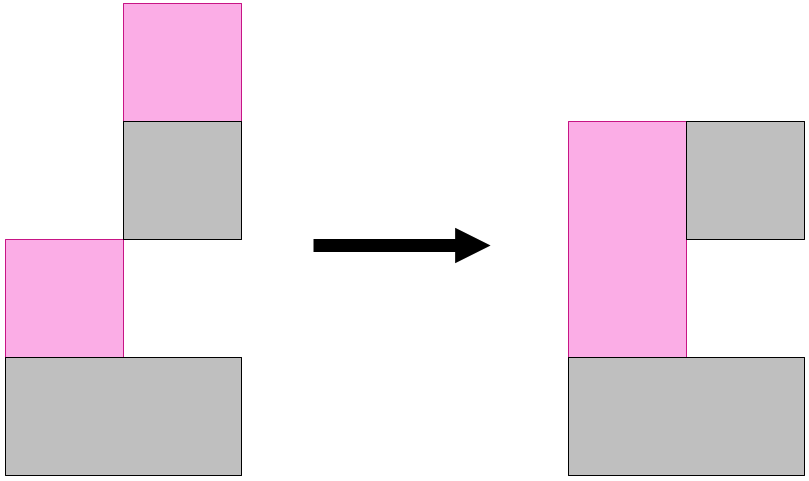}
                 \subcaption[]{The merging mechanism}
                 \label{jelfusion}
            \end{subfigure}
            \caption{The game of Jelly-No \cite{Qrostar}}
            \label{fig:intro_jelly}
       \end{figure}
        The goal of the game is to merge all jellies of similar colour into one single jelly. The detailed rules can be found in \cite{JellyNo}.

        Hanano (see Figure \ref{fig:intro_hanano}) contains coloured movable blocks of size $1\times 1$, coloured flowers represented by blocks of size $1\times 1$ marked with a flower pattern and movable grey blocks of any size or shape. See for example Figure \ref{hanexample}.
        In this game, the player can either move a coloured block, or horizontally swap two adjacent blocks. Coloured blocks are marked with an arrow: if a block touches a flower of the same colour, and if there is enough space to allow it, this block blooms and a new flower appears, attached to it, in the direction of the arrow. Note that this new flower can push blocks around it in order to make enough space to bloom. See Figure \ref{hanbloom} where flowers are marked with stripes and the letter ``F", light grey rectangles represent fixed platforms and dark grey rectangles represent grey movable blocks. If the new flower touches a block of the same colour, this block also blooms, creating a chain reaction. 
       \begin{figure}[H]
         \begin{subfigure}[t]{0.6\textwidth}
            \centering
             \includegraphics[width=0.7\textwidth]{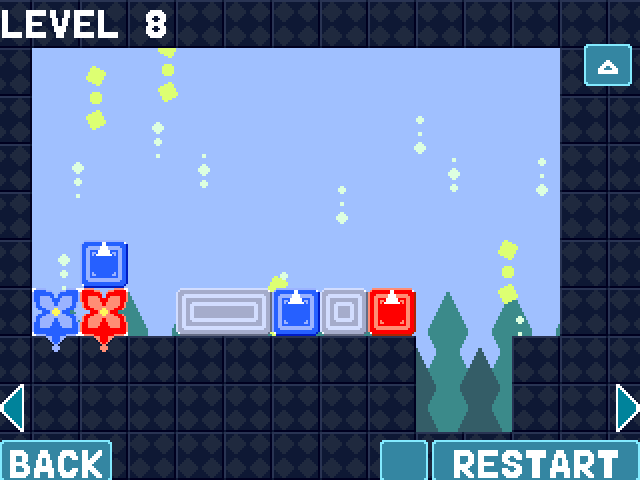}
             \subcaption[]{Level 8 of Hanano, with movable grey blocks}
             \label{hanexample}
            \end{subfigure}
            \begin{subfigure}[t]{0.4\textwidth}
                \centering
                 \includegraphics[width=0.9\textwidth]{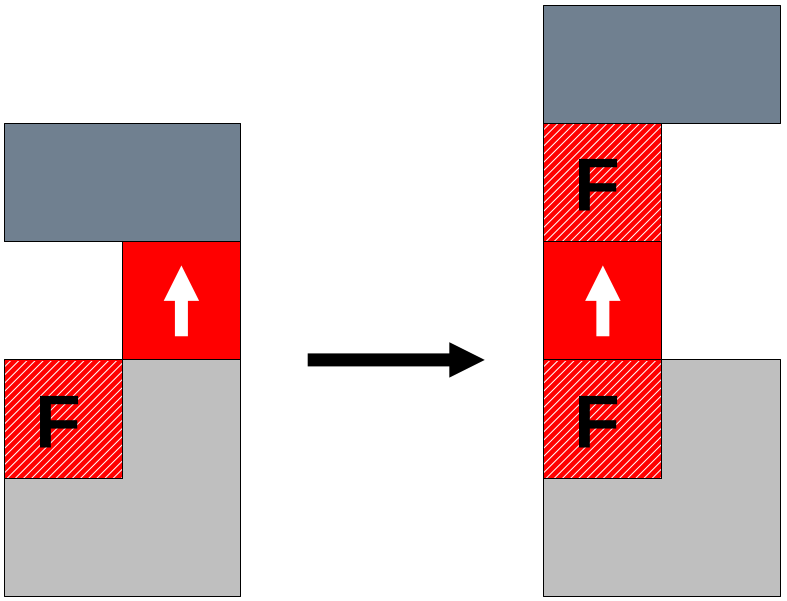}
                 \subcaption[]{The blooming mechanism}
                 \label{hanbloom}
            \end{subfigure}
            \caption{The game of Hanano \cite{Qrostar}.}
            \label{fig:intro_hanano}
       \end{figure}
       Similarly to Jelly-No, movable blocks are subject to gravity.        
       Coloured blocks and flowers can be red, blue or yellow. Flowers are all attached to a block or platform and move if, and only if, the block they are attached to moves. 
        The goal of the game is to make all the coloured blocks bloom. The detailed rules can be found in \cite{HananoJ}.

    \subsection{Problems studied and state of the art}
    
        In this work, we study the complexity of the two following problems under various constraints:
    
        \textsc{Jelly} = \{ $J$ | $J$ is a solvable level of Jelly-No \}

        \textsc{Hanano} = \{ $H$ | $H$ is a solvable level of Hanano \}.

        Constraints can affect either the size (height or width) of the board or the number of colours. Restrictions on the size, unless very strong, do not necessarily lower the complexity classes of games: for example, the famous puzzle game Tetris was shown by Sualeh Asif, Michael Coulombe, Erik D.Demaine, Martin L. Demaine, Adam Hesterberg, Jayson Lynch and Mihir Singhal \cite{Tetris} to be NP-hard even when restricting the size of the board to 8 columns or 4 rows, although it turns out to be polynomial for 2 columns or 1 row. When restricting \textsc{Jelly} to 1 colour, we assume that we cannot use black jellies unless stated otherwise. Indeed, using $n$ black jellies is equivalent to using $n$
jellies each of a distinct colour, strongly increasing game complexity. 
    
        \textsc{Jelly} was shown by Chao Yang \cite{JellyNo} to be in NP and NP-hard under the restriction that all movable blocks have the same colour. Note that a playable version of the reduction presented in \cite{JellyNo} has been implemented by Giorgio Ciotti \cite{IsNPHard}, allowing for a better understanding of the reduction.
 However, PSPACE-completeness for more than one colours was left as an open question. Evidence towards the fact that the game for more than one colours is more complex was already presented in \cite{JellyNo} where the author shows the existence of levels that require an exponential number of steps to solve.
        
        \textsc{Hanano} was shown by Ziwen Liu and Chao Yang \cite{HananoNPHard} to be NP-hard even under the restrictions that (1) all movable blocks and flowers have the same colour (2) there are no movable grey blocks and (3) blocks can only bloom upwards. It was then shown in a later work by Michael C. Chavrimootoo \cite{HananoJ} to be PSPACE-complete under the restrictions that (1) all movable blocks and flowers have the same colour and (2) blocks can only bloom upwards.
        An important result of this paper, which we reuse in this work, was to circumvent the effects of gravity in \textsc{Hanano} making some moves irreversible.
        
    \subsection{Contributions}

        The rest of this work is organised as follows:
        
        \begin{enumerate}
            \item {In Section \ref{generalJelly}, we show that \textsc{Jelly} is PSPACE-hard if we allow to use an unbounded number of colours, and strengthen this result by adapting the reduction so that we only need one single colour and the use of black jellies, providing new results on an open question raised by Chao Yang in \cite{JellyNo}. We then show that \textsc{Jelly} is in PSPACE in the general case, resulting in PSPACE-completeness.}
            \item {In Section \ref{restrictedJelly}, we study \textsc{Jelly} with constraints on one dimension of the board and the number of colours. We show, by reduction from the \textsc{3-Partition} problem, that the following are NP-hard: 1-Colour \textsc{Jelly} where the height of the board is limited to nine lines, 2-Colour \textsc{Jelly} where the height is limited to three lines and 1-Colour \textsc{Jelly} where the width of the board is limited to five columns. }
            \item {In Section \ref{restrictedHanano}, we study \textsc{Hanano} with similar constraints. We show that 1-Colour \textsc{Hanano} is still NP-hard when the length of the board is limited to six columns or when the height is limited to eleven lines.}
        \end{enumerate}
        

\section{JELLY with black jellies and no bound on the size of the board}
\label{generalJelly}

    We study the complexity of \textsc{Jelly} with no bound on the size of the board. Our main result is that allowing black jellies in 1-Colour \textsc{Jelly} renders the game PSPACE-complete.
    
    We prove first that it is PSPACE-complete with an unbounded number of colours. Then, we strengthen the result to 1-Colour \textsc{Jelly} where black jellies are allowed.

    \subsection{Previous works}

    To prove PSPACE-hardness in both cases, we construct a reduction from a variation of the Nondeterministic Constraint Logic (NCL) problem, a graph problem presented and studied by Robert A. Hearn and Erik D. Demaine \cite{ncl}. This problem was used by Michael C. Chavrimootoo to study \textsc{Hanano} in \cite{HananoJ} and is very useful to prove PSPACE-hardness of games of sliding or pushing blocks. 

    Our reduction borrows many ideas from the work of Michael C. Chavrimootoo \cite{HananoJ}, notably the use of visibility representations and the idea of representing vertices as vertex gadgets to be solved and flows as blocks (in our case jellies) travelling between the vertex gadgets through horizontal tunnels.

    A major difference between the two games is that in \textsc{Hanano}, the blooming mechanism allows to push a block upward by creating a new flower block. In \textsc{Jelly} (at least for the simplest version of the game that we consider in this work), there is no way for a jelly to go upward and no new block of jelly can be created. Besides, in \textsc{Jelly}, solving a level requires bringing together all jellies of the same colour. All these differences required to adapt the ideas presented in \cite{HananoJ} in order to build suitable gadgets.

    
    \subsection{JELLY with an unbounded number of colours}
    
    The fact that \textsc{Jelly} is in PSPACE can be proven by using standard arguments. The proof can be found in Section \ref{sec:inPspace} (Theorem \ref{inPspace}). In this paragraph and the next one, we shall focus our efforts in proving that the game is in fact PSPACE-hard.
    
    First, we study the complexity of \textsc{Jelly} with an unbounded number of different colours for the jellies. We prove the following result:
    
    \begin{theorem}
        \label{JellyPHard}
        \textsc{Jelly} with an unbounded number of colours is PSPACE-hard.
    \end{theorem}

        
        \subsubsection{NCL graph and visibility representation}
    
        An NCL graph is a directed graph that meets the following constraints:
        \begin{enumerate}
            \item Edges are of weight either one (red edges) or two (blue edges).
            \item For every vertex, the incoming flow is at least two (we call this constraint the minimum inflow requirement).
        \end{enumerate}
        Given such graph $G = (V, E)$, given an edge $e \in E$, the problem is deciding whether there exists a sequence of edge direction flips such that the direction of $e$ can be flipped and all intermediary graphs as well as the resulting graph are still NCL. Another way to view this problem, which is the one used in \cite{ncl}, is to consider the graph as undirected and each assignment of directions as a configuration: moving from one configuration to another means flipping the direction of an edge such that the resulting configuration still meets the incoming flow requirement constraint. The problem would then be to find a sequence of configurations that result in $e$ being flipped. This problem happens to be PSPACE-complete even with the following additional constraints \cite{ncl} (see Figure \ref{NCL}):
        \begin{enumerate}  
            \setcounter{enumi}{2}
            \item All vertices have exactly three incident edges.
            \item Vertices follow two possible patterns: AND (two incident red edges and one incident blue edge) or OR (all three incident edges are blue) (see Figure \ref{AndOr}).
            \item The graph is simple (without loops or multiple edges) and planar (it can be drawn in two dimensions without any of its edges crossing).
        \end{enumerate}
        \begin{figure}[H]
            \begin{minipage}[t]{0.5\textwidth}
                \centering
                 \includegraphics[width=0.6\textwidth]{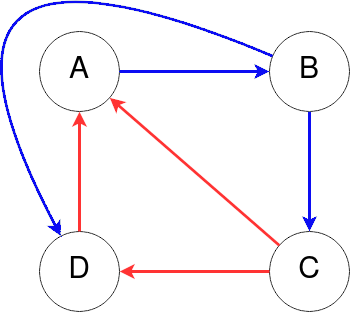}
                 \caption[]{Graph respecting our constraints}
                 \label{NCL}
             \end{minipage}
             \begin{minipage}[t]{0.5\textwidth}
                \centering
                 \includegraphics[width=\textwidth]{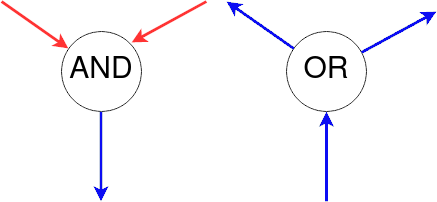}
                 \caption[]{OR and AND patterns}
                \label{AndOr}
            \end{minipage}
        \end{figure}
    
        In order to reduce the NCL problem to \textsc{Jelly}, all we need to do is construct gadgets to emulate AND and OR vertices as well as a mechanism for flipping edges and ensure that all the constraints are respected. However, due to gravity, some moves in \textsc{Jelly} are irreversible, which is not the case in the NCL problem. The problem of circumventing the effects of gravity in \textsc{Hanano} was treated very efficiently in \cite{HananoJ} by using the notion of visibility representation: a graph representation for planar graphs studied notably by Roberto Tamassia and Ioannis G.Tollis, who presented definitions as well as linear-time algorithms for constructing such representations \cite{Tamassia1986}. The idea used in \cite{HananoJ}, which we reuse here, was to work not from the graph $G$ but from its visibility representation.
    
        Given a planar graph $G$, a visibility representation is a 2D representation of $G$ that maps vertices to vertical segments and edges to horizontal ones (possibly arrows if $G$ is directed) such that if two vertices are linked by an edge, then the corresponding vertical segments are also linked by a horizontal segment. Otherwise, segments do not intersect each other. See example on Figure \ref{visibrep}.
    
        \begin{figure}[H]
            \centering
            \includegraphics[width=0.38\textwidth]{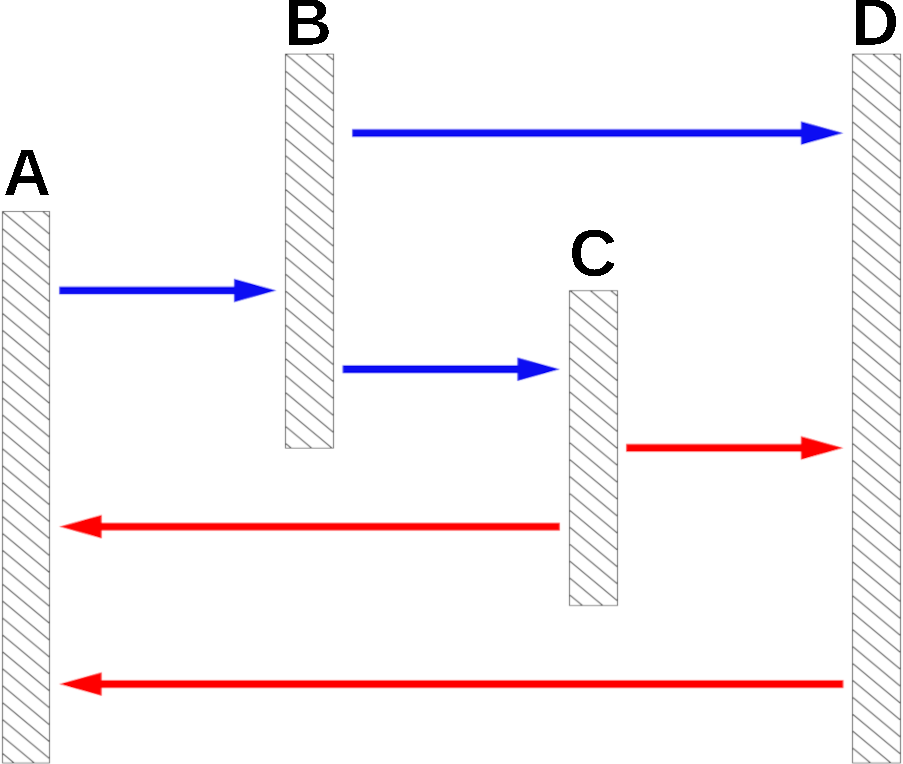}
            \caption{Visibility representation corresponding to graph shown in Figure \ref{NCL}, with vertices represented by rectangles instead of vertical segments for a better visibility.}
            \label{visibrep}
        \end{figure}
    
        For the sake of our reduction, we also assume that two horizontal segments do not share $y$ coordinates. We can note that in \cite{Tamassia1986} the vertices were represented by horizontal segments and the edges by vertical ones, but the definition used in our work, which is also the one used in \cite{HananoJ}, is equivalent and more adapted to the purpose of our reduction.
        
        Given the visibility representation for $G$, each vertex can be annotated with the edges to which it is connected and their position relatively to the other incident edges: left or right and top, middle or bottom. We denote by ``B" (resp. ``R") segments representing blue (resp. red) edges, with over arrows to indicate the side (left or right) to which the edge connects to the vertex. 
        For example, in Figure \ref{visibrep}, vertex $A$ can be represented by ($\overrightarrow{B}, \overrightarrow{R}, \overrightarrow{R}$), meaning it has no incident edge on its left side, a blue edge at the top of its right side and red edges at the middle and bottom of its right side; vertex $B$ can be represented by ($\overrightarrow{B}, \overleftarrow{B}, \overrightarrow{B}$), vertex $C$ by ($\overleftarrow{B}, \overrightarrow{R}, \overleftarrow{R}$) and vertex $D$ by ($\overleftarrow{B}, \overleftarrow{R}, \overleftarrow{R}$). Note that the direction of each over arrow does not correspond to the direction of the corresponding edge, but to the side, which depends only on vertex order in the visibility representation. A change of edge direction does not change this side.
        This notation, similar to the one used in \cite{HananoJ}, allows us to list all the possible configurations for vertices in a visibility representation.
        
        With these tools and notations in mind, we can now build our reduction.


        \subsubsection{Reduction from NCL to \textsc{Jelly}}
        
        Let us first present a global idea of the reduction through an example, before detailing the construction.
        
        Starting from a given NCL graph $G = (V, E)$ (see Figure \ref{NCL}), we first compute the corresponding visibility representation.

        Each vertex is represented by a gadget which we call ``vertex gadget".
        Due to the structure used, instead of only having to build gadgets for AND and OR patterns, we now need to show existence of vertex gadgets corresponding to each possible configuration of adjacent edges in the visibility representation. However, results from \cite{HananoJ} proved that three such gadgets are enough to build all remaining configurations. The first two gadgets should correspond to an OR and an AND patterns: let us build for example $(\overleftarrow{B}, \overrightarrow{B}, \overleftarrow{B})$ for OR and $(\overleftarrow{B}, \overleftarrow{R}, \overrightarrow{R})$ for AND. The third gadget is called ``red (resp. blue) bend". It is similar to a vertex gadget, except it only has two incident edges on the same side, and its solving zone slightly differs. We use it to change the direction of a tunnel in a visibility representation. Figures \ref{coloured} and \ref{fig:visibility_to_reduction} show a direct visibility representation for Figure \ref{NCL} and its corresponding representation using only configurations $(\overleftarrow{B}, \overrightarrow{B}, \overleftarrow{B})$, $(\overleftarrow{B}, \overleftarrow{R}, \overrightarrow{R})$ and red or blue bends. We obtain vertices $A$ and $D$ by combining $(\overleftarrow{B}, \overleftarrow{R}, \overrightarrow{R})$ gadgets and red bends, and $C$ by combining a $(\overleftarrow{B}, \overrightarrow{B}, \overleftarrow{B})$ gadget and a blue bend.
        
        \begin{figure}[H]
            \begin{minipage}[t]{0.4\textwidth}
                \centering
                \includegraphics[width=\textwidth]{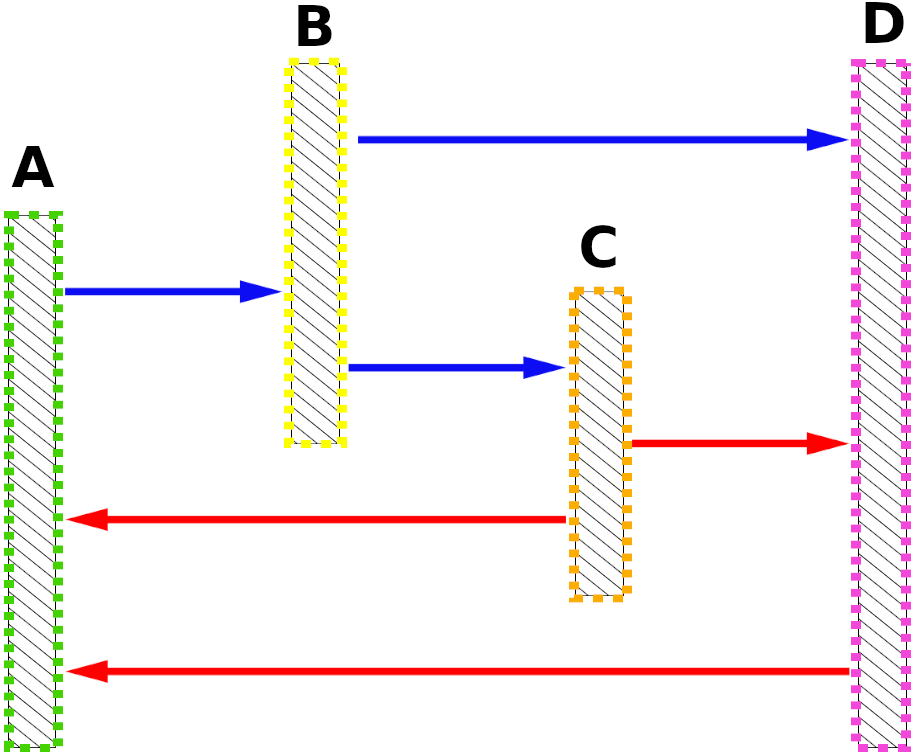}
                \caption{Visibility representation for Figure \ref{NCL}. Each vertex has a distinct colour.}
                \label{coloured}
            \end{minipage}
            \begin{minipage}[t]{0.02\textwidth}
                \;\;
            \end{minipage}
            \begin{minipage}[t]{0.57\textwidth}
                \centering
                \includegraphics[width=\textwidth]{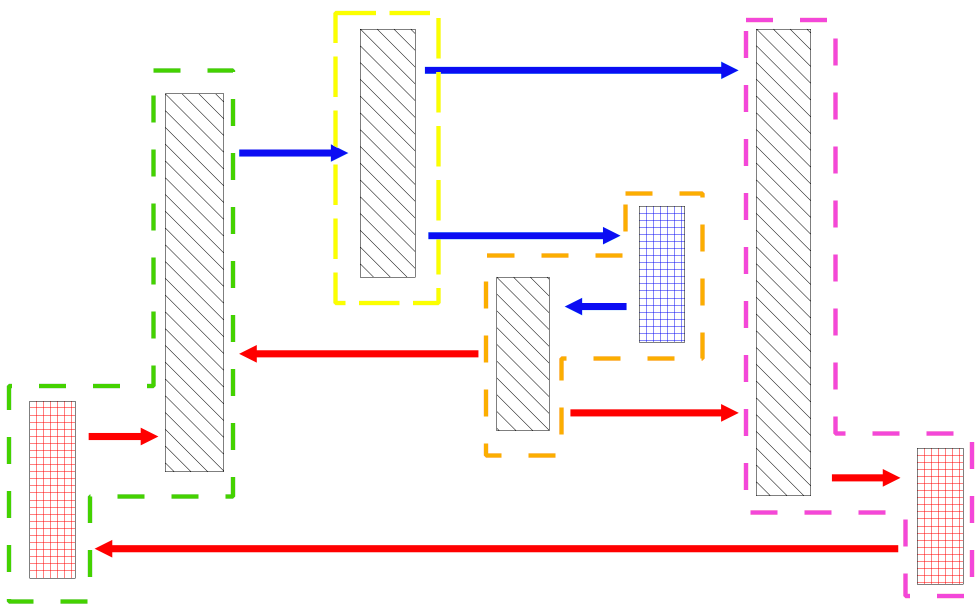}
                \caption{Visibility representation for Figure \ref{NCL} with restricted configurations. Chequered rectangles represent red (resp. blue) bends.}
                \label{fig:visibility_to_reduction}
            \end{minipage}
        \end{figure}
       
        For each edge $e = (u, v) \in E$, we link the two gadgets for $u$ and $v$ by a horizontal tunnel with a jelly of width 1 or 2 depending on the flow in $e$. 
        If the jelly corresponding to edge $e$ is in the vertex gadget for $u$, it corresponds to $e$ being oriented toward $u$ in the graph.
        
        Let $e = (A, B)$ be the edge which we want to flip to $(B, A)$. We will sometimes call $A$ the ``target vertex". We add a fixed, extra jelly of the same colour as the one corresponding to $e$ to the bottom of the gadget corresponding to $A$ to enforce the condition that $e$ must be flipped. This representation gives us the following level (Figure \ref{fullLevel}), where the jelly for $e$ and the extra jelly are circled in blue and each vertex gadget has the colour corresponding to the one in Figure \ref{coloured}.
    
        \begin{figure}[H]
            \centering
            \includegraphics[width=0.9\textwidth]{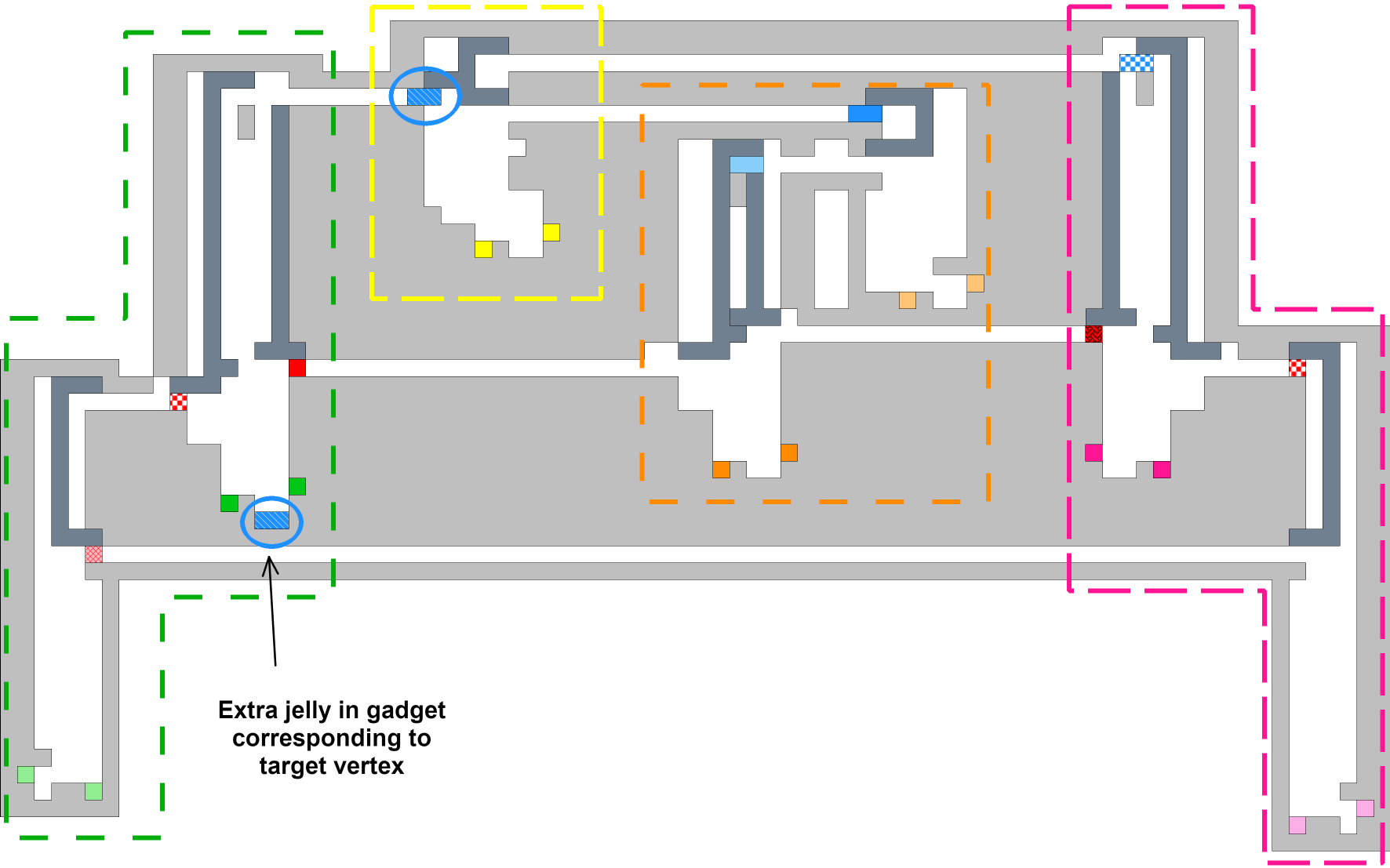}
            \caption{Reduction corresponding to Figure \ref{fig:visibility_to_reduction}. Each dotted form corresponds to a vertex in the original graph from Figure \ref{NCL}.}
            \label{fullLevel}
        \end{figure}
        
        Let us now detail the proof of Theorem \ref{JellyPHard}.

        \begin{proof}
        
        We consider the following problem:
        
        \textsc{Instance:} Let $G = (E, V)$ be a simple planar NCL graph such that all vertices have exactly three incident edges and follow either an OR or AND pattern. Let $e = (u, v) \in E$. 
        
        \textsc{Question:} Is there a succession of edge flips such that $e$ can be flipped and all intermediate steps, as well as the resulting graph, are still NCL?
        
        To reduce our problem to \textsc{Jelly}, we begin by computing the visibility representation corresponding to $G$, which can be done in linear time \cite{Tamassia1986}. 
        Using results from \cite{HananoJ}, we compute the equivalent representation using only configurations $(\overleftarrow{B}, \overrightarrow{B}, \overleftarrow{B})$, $(\overleftarrow{B}, \overleftarrow{R}, \overrightarrow{R})$ and red or blue bends. We obtain each configuration with a constant number of combinations, thus in linear time.

        In the rest of the paragraph, we name by ``vertices" and ``edges" the corresponding vertical shapes and horizontal segments in the representation, with no distinction between OR/AND patterns and red or blue bends. Let $V', E'$ be the sets of ``vertices" and ``edges" in the representation.
        Since our version of \textsc{Jelly} has an unbounded number of colours, we can assign one colour to each vertex and one colour to each edge of our representation. 
        Let $n_i, b_i^1, b_i^2, b_i^3, i \in [1, |V'|]$ and $m_j, j \in [1, |E'|]$ be the different colours used in our reduction.

        We build vertex gadgets corresponding to $(\overleftarrow{B}, \overrightarrow{B}, \overleftarrow{B})$ and $(\overleftarrow{B}, \overleftarrow{R}, \overrightarrow{R})$ configurations and red or blue bends. Those gadgets are linked by tunnels of height 1, which correspond to the edges.

        In tunnels corresponding to blue edges, we place a jelly of width 2 and height 1, which we call a ``blue edge jelly" (or ``blue jelly" for short), and similarly in tunnels corresponding to red edges, a ``red (edge) jelly" of size $1\times1$, such that each jelly corresponding to edge $e_j$ has a unique colour $m_j$. Another option, equivalent to using distinct colours, is to use black jellies which cannot merge. These edge jellies can move from one incident vertex gadget to the other through their assigned tunnel, representing a change of direction of the corresponding edge in the graph. 
        
        For readability purposes we represent blue (resp. red) edge jellies in blue (resp. red) on the figures, but with variations of patterns to indicate that in the reduction, they should all either have distinct colours or be black jellies. Our constructions also use these black jellies, which we represent in dark grey for readability purposes. Light grey blocks represent fixed, unmovable walls and platforms. Note that we could also represent these black jellies using unique colours $b_i^1, b_i^2, b_i^3$ for each, therefore we do not strictly need to use this feature. We use the property mentioned in Section \ref{gameplay} that black jellies can have any shape or size an ordinary jelly could have, and not necessarily rectangular or square.

        For each gadget, we must make sure that:
        \begin{enumerate}
            \item {Edge jellies can only access the two vertex gadgets incident to their assigned tunnel, and no other jelly can travel through this tunnel.}
            \item{At every step of the game, each gadget must contain at least one blue edge jelly or two red edge jellies; otherwise the puzzle becomes unsolvable.}
            \item{As long as the part of the puzzle corresponding to some vertex gadget is not solved and the previous constraint is respected, moving the edge jellies in and out of the gadget must be reversible.}
            \item{The height of the vertex gadgets, especially the vertical space between the tunnels, should easily be modified to fit the visibility representation, where vertices have different heights to make sure the tunnels do not cross.}
        \end{enumerate}
        
        All gadgets use a similar idea: we place the entries for all three tunnels at the corresponding position in the vertex gadget and the corresponding edge jelly in this tunnel. At the bottom of the gadget, there are two jellies of the same colour $n_i$ (specific to each vertex) separated by a gap of width 2 and height 1, which we call ``vertex jellies". One of these jellies, on the right, is fixed inside the wall so it cannot move. The other, on the left, can move. The only way for the jellies to merge is for the gap to be filled by using either one blue edge jelly or two red edge jellies. Black jellies ensure that each gadget respects the constraints at any given time.

        Let us now consider the edge jelly corresponding to edge $e$ chosen at the beginning of the game, which we want to flip from $(u,v)$ to $(v,u)$. This corresponds in our reduction to forcing the edge jelly to be in the gadget corresponding to target vertex $u$ at the end of the game. To ensure this constraint, we add a jelly of the same colour in the ground of the gap at the bottom of the gadget corresponding to $u$. Therefore the only way to complete the level is to stick the edge jelly at the bottom of the gadget, which corresponds to fixing the direction of the edge.

        We start by constructing the OR vertex gadget, corresponding to configuration $(\overleftarrow{B}, \overrightarrow{B}, \overleftarrow{B})$. The gadget is visible in Figure \ref{fig:or}.
        In this construction, one blue jelly can support the black jelly alone. If a blue jelly falls either in or under it, it can still access the gap separating the vertex jellies and solve the gadget even without support from the other blue jellies (see Figures \ref{fig:or_solved1} and \ref{fig:or_solved2}). Otherwise, the black jelly blocks access to the gap and the puzzle cannot be solved. This ensures the minimum inflow requirement.

        \begin{figure}[H]
            \begin{minipage}[t]{0.3\textwidth}
                \centering
                \includegraphics[width=\textwidth]{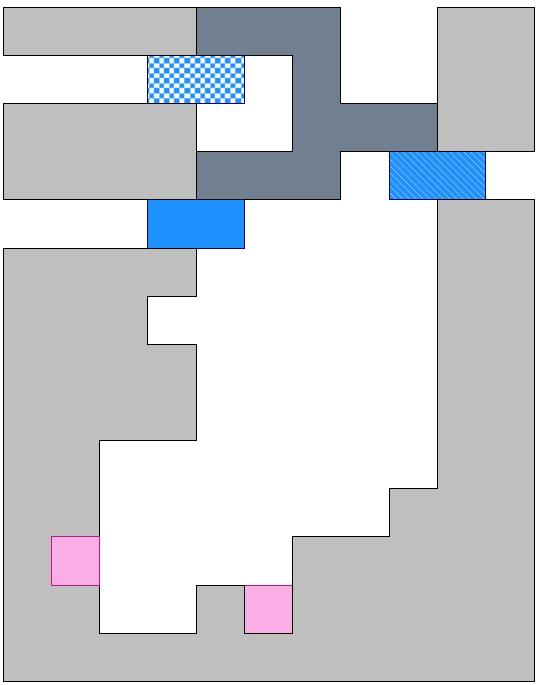}
                \caption{$(\overleftarrow{B}, \overrightarrow{B}, \overleftarrow{B})$ vertex gadget.}
                \label{fig:or}
            \end{minipage}
            \begin{minipage}[t]{0.03\textwidth}
                \;\;
            \end{minipage}
            \begin{minipage}[t]{0.3\textwidth}
                \centering
                \includegraphics[width=\textwidth]{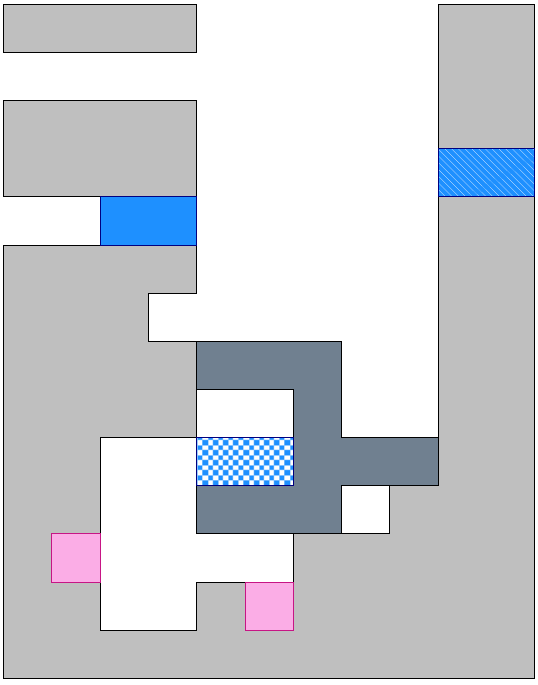}
                \caption{First solution using top blue jelly.}
                \label{fig:or_solved1}
            \end{minipage}
            \begin{minipage}[t]{0.03\textwidth}
                \;\;
            \end{minipage}
            \begin{minipage}[t]{0.3\textwidth}
                \centering
                \includegraphics[width=\textwidth]{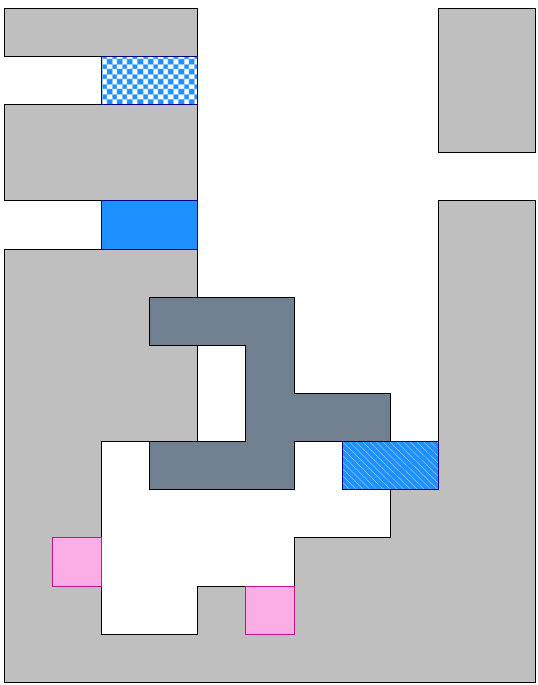}
                \caption{Second solution using middle or bottom jelly.}
                \label{fig:or_solved2}
            \end{minipage}
        \end{figure}

        The next step in the reduction is constructing the AND vertex gadget, corresponding to configuration $(\overleftarrow{B}, \overleftarrow{R}, \overrightarrow{R})$. 
        We must ensure that at any given time, the puzzle is solvable if, and only if, at least one blue or two red jellies are present in the gadget.
        Figure \ref{fig:and} presents this gadget. We use two black jellies, which can either be supported by either two red jellies, or a blue jelly alone as visible in Figure \ref{fig:and_blue}.
        The red jellies can slide under black jellies to fill the bottom gap. To solve the gadget with a blue jelly instead, one can move the black jellies to the sides as long as the blue jelly stays in place, as visible on Figure \ref{fig:and_solved}. Once again, if the blue jelly and at least one of the red jellies leave, the black jellies fall and block access to the gap. This ensures the minimum inflow requirement.

        \begin{figure}[H]
            \begin{minipage}[t]{0.3\textwidth}
                \centering
                \includegraphics[width=\textwidth]{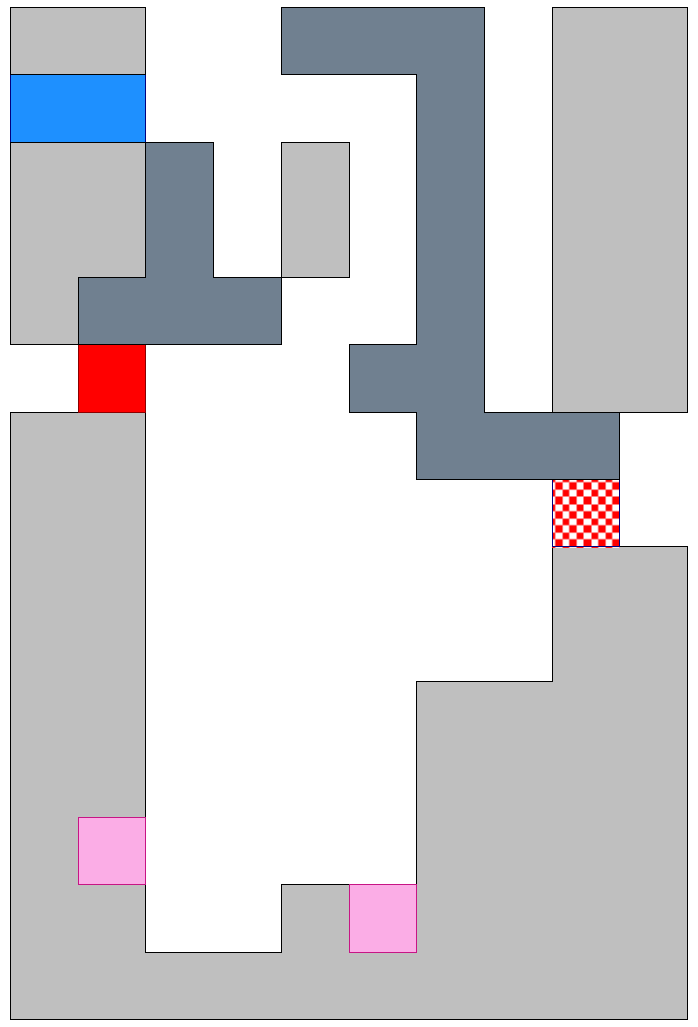}
                \caption{$(\overleftarrow{B}, \overleftarrow{R}, \overrightarrow{R})$ vertex gadget.}
                \label{fig:and}
            \end{minipage}
            \begin{minipage}[t]{0.03\textwidth}
                \;\;
            \end{minipage}
            \begin{minipage}[t]{0.3\textwidth}
                \centering
                \includegraphics[width=\textwidth]{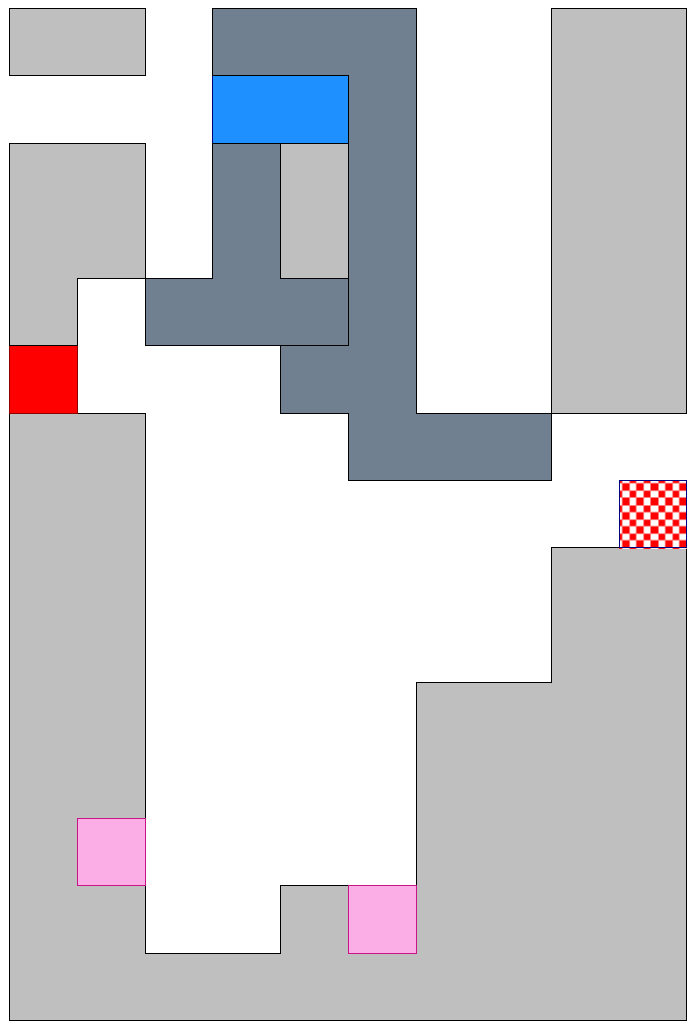}
                \caption{Holding black jellies with blue jelly.}
                \label{fig:and_blue}
            \end{minipage}
            \begin{minipage}[t]{0.03\textwidth}
                \;\;
            \end{minipage}
            \begin{minipage}[t]{0.3\textwidth}
                \centering
                \includegraphics[width=\textwidth]{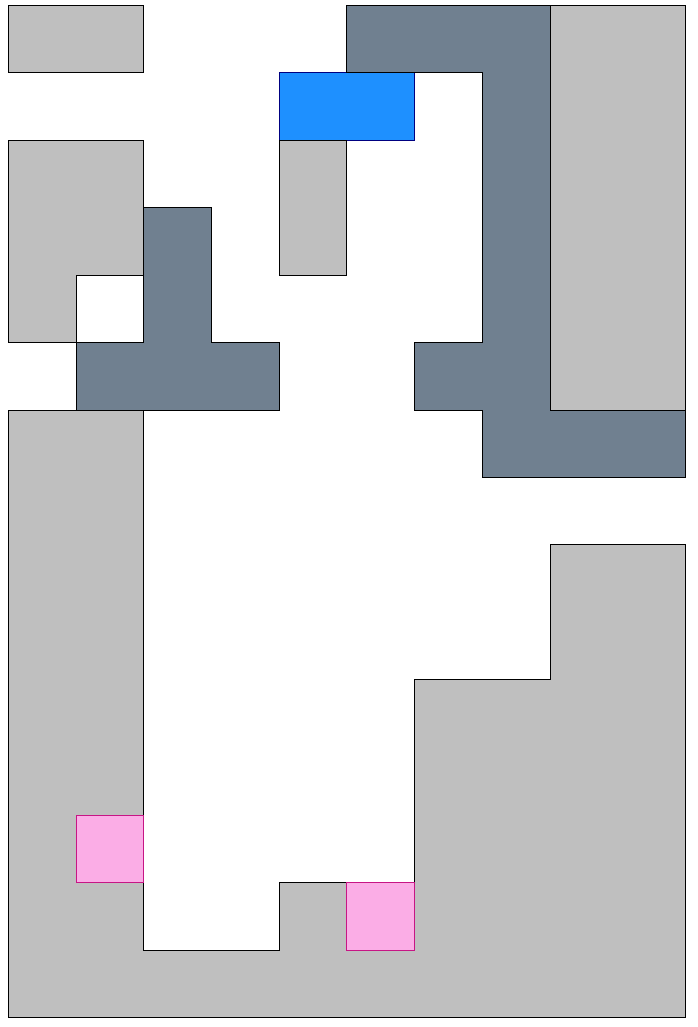}
                \caption{Blue jelly can access the bottom gap.}
                \label{fig:and_solved}
            \end{minipage}
        \end{figure}

        Finally, we construct the red (or blue) bend gadget, which allows the construction of all remaining configurations following results from \cite{HananoJ}. This gadget only has two tunnels on the same side, each with the same flow, i.e. 1 for red bend and 2 for blue bend. Besides, the red bend differs from other vertex gadgets in that solving it requires only one red jelly (and one blue jelly for a blue bend). We present the red bend gadget in Figure \ref{fig:red_bend} as well as its solution and its blocked state, when the minimum inflow requirement is not met. 

        \begin{figure}[H]
            \begin{subfigure}[t]{0.35\textwidth}
                \centering
                \includegraphics[width=0.5\textwidth]{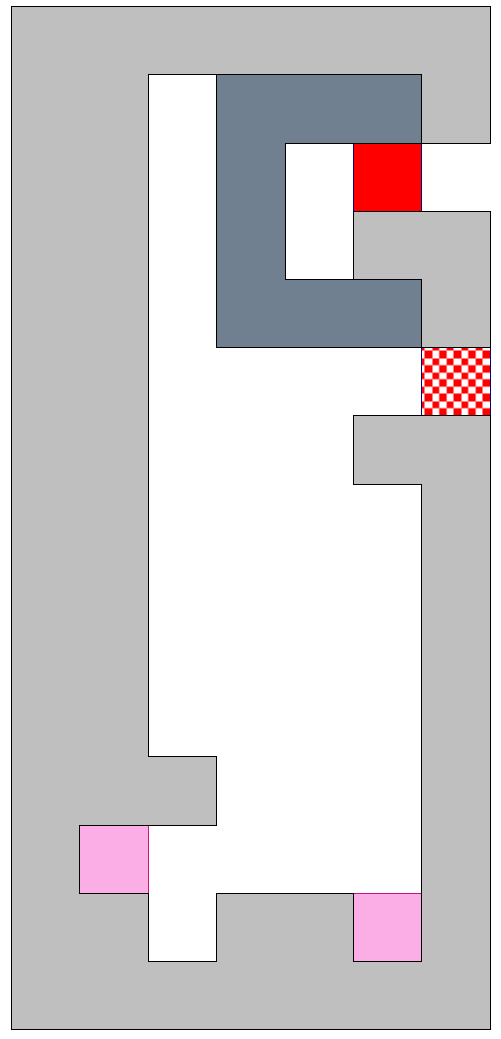}
                \subcaption{Red jelly supporting black jelly.}
            \end{subfigure}
            \begin{subfigure}{0.03\textwidth}
                \;\;
            \end{subfigure}
            \begin{subfigure}[t]{0.27\textwidth}
                \centering
                \includegraphics[width=0.65\textwidth]{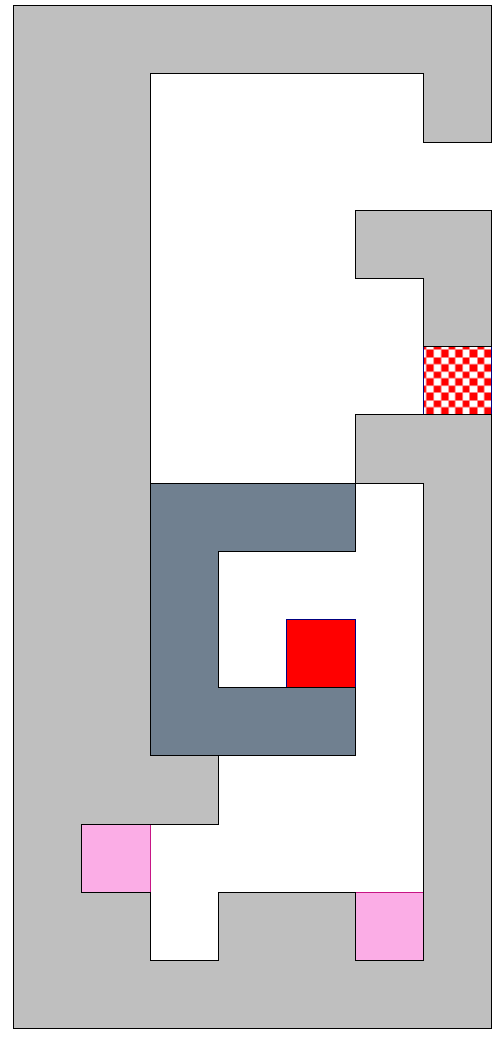}
                \subcaption{Solving gadget.}
            \end{subfigure}
            \begin{subfigure}{0.03\textwidth}
                \;\;
            \end{subfigure}
            \begin{subfigure}[t]{0.27\textwidth}
                \centering
                \includegraphics[width=0.65\textwidth]{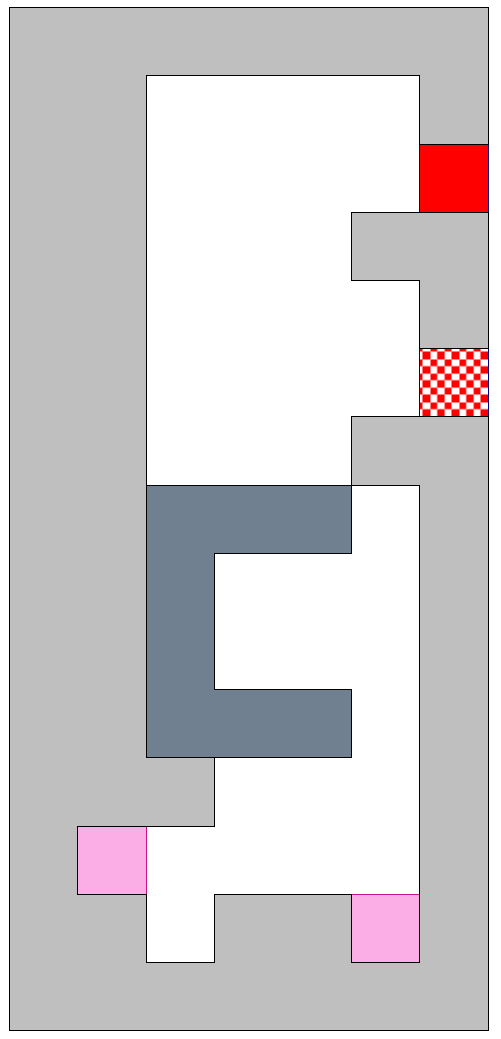}
                \subcaption{Blocked gadget.}
            \end{subfigure}
            \caption{Red bend gadget.}
            \label{fig:red_bend}
        \end{figure}

        This gadget is equivalent to reorienting an edge, from left to right or from middle to top/bottom.
        Whenever a configuration is obtained by combining several, we may need to move several jellies between combined gadgets in order to preserve flow, but once again the number of extra moves is constant for each configuration. For example, in Figure \ref{fig:example_bend_flow}, we combine $(\overleftarrow{B}, \overleftarrow{R}, \overrightarrow{R})$  and a blue bend to obtain $(\overrightarrow{B}, \overleftarrow{R}, \overrightarrow{R})$. To move the blue jelly outside of the gadget, we move it to the blue bend. This allows the blue jelly from the blue bend to leave while respecting the minimum inflow requirement.

        \begin{figure}[H]
            \begin{subfigure}{0.49\textwidth}
                \centering
                \includegraphics[width=0.9\textwidth]{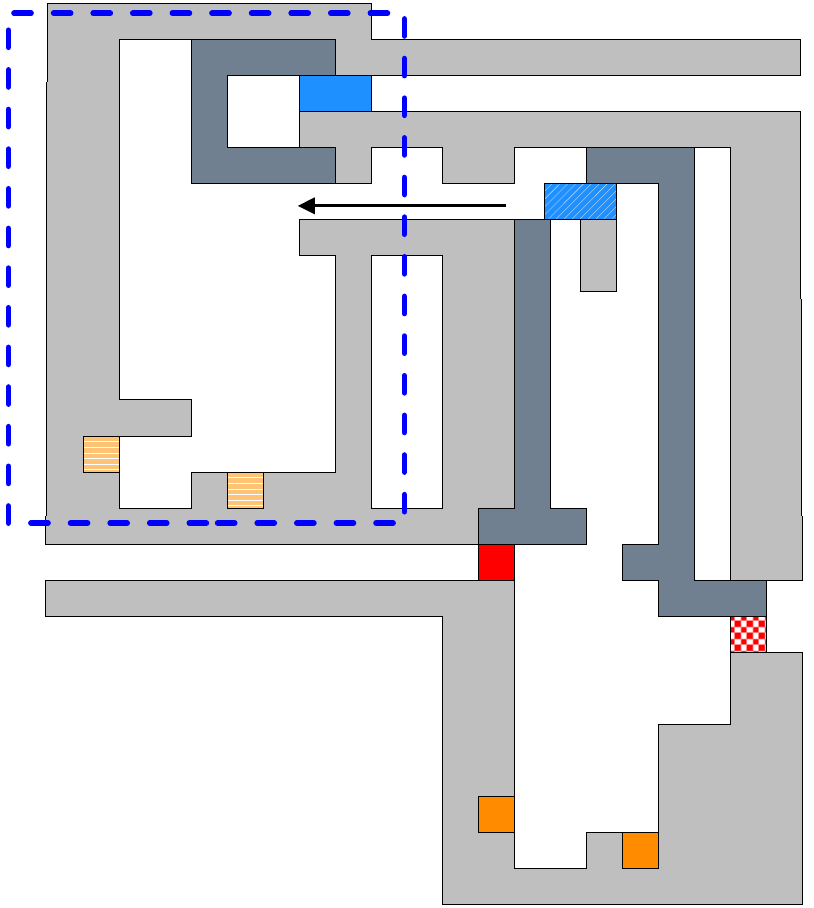}
                \subcaption{First jelly moves from $(\overleftarrow{B}, \overleftarrow{R}, \overrightarrow{R})$ to blue bend.}
            \end{subfigure}
            \begin{subfigure}{0.02\textwidth}
                \;\;
            \end{subfigure}
            \begin{subfigure}{0.49\textwidth}
                \centering
                \includegraphics[width=0.9\textwidth]{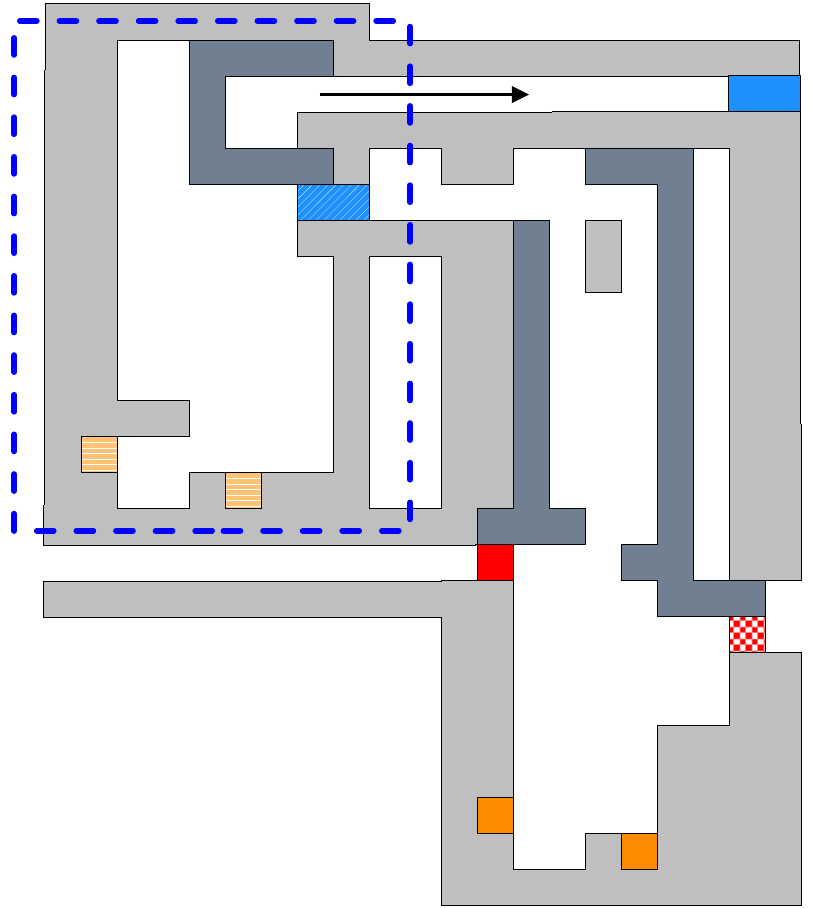}
                \subcaption{Second jelly leaves blue bend.}
            \end{subfigure}
            \caption{Flow reorientation in $(\overrightarrow{B}, \overleftarrow{R}, \overrightarrow{R})$ built from $(\overleftarrow{B}, \overleftarrow{R}, \overrightarrow{R})$  and a blue bend.}
            \label{fig:example_bend_flow}
        \end{figure}

        Due to gravity and the positions of the black jellies, no jelly can move into a tunnel it was not assigned to (i.e. a blue jelly in a red tunnel or the other way round, or two jellies in the same tunnel). Besides, we can easily check the other constraints for all gadgets by checking that both resolutions are possible (one blue jelly or two red jellies) and that if the constraints are not respected, the gadget cannot be solved, either because black jellies prevent edge jellies from entering the vertex gadget or because edge jellies fall to the bottom and prevent the vertex jellies from merging.
                  
         It is easy to see that if our instance of restricted NCL has a solution, i.e. if given an edge $e \in E$ there exists a succession of edge flips such that the direction of $e$ can be flipped while always respecting the constraints, then the level of \textsc{Jelly} resulting from our reduction also has a solution. For each edge flip $(v_i, u_i)$ to $(u_i, v_i)$ in the solution, the corresponding edge jelly, initially in the vertex gadget corresponding to $u_i$, moves to the one corresponding to $v_i$. If at every step of this succession of flips the minimum inflow requirement is still met, then the construction of each gadget stays in place. Eventually, one can solve all vertex gadgets by filling the gaps with either one blue or two red jellies and merge the jelly corresponding to $e$ with the corresponding jelly in vertex gadget corresponding to $u$. 

         Now, we must ensure that if the level of \textsc{Jelly} constructed by our reduction is solvable, then the corresponding instance of NCL has a solution.
         Due to the construction of the gadgets, the only way for the game to be solved is if (1) the gap in every gadget is filled by either one blue jelly or two red jellies and (2) the jelly corresponding to $e$ can merge with the jelly in the gadget corresponding to $u$. 
         This corresponds exactly to
         (1) meeting the minimum inflow requirement constraint in every vertex gadget and at every step of the game (if there exists a move and a vertex gadget for which the requirement is violated then the construction in this gadget falls apart and the black jellies fall down, thus the edge jellies cannot reach the gap and allow vertex jellies to merge, making the level impossible to finish) and
         (2) flipping the direction of edge $e$, which was the goal to reach.
         
        Therefore the original problem has a solution, which can be determined by following the succession of moves leading to a winning configuration of the level.
        
        This concludes our proof of Theorem \ref{JellyPHard}.
        \end{proof}

        Note that if we use black jellies for all the edges other than $e$, then we only need |V| + 1 colours for our reduction: one for each vertex of $G$ and one for $e$. If we forbid the use of black jellies, then since vertex gadgets use three extra jellies at most, we need to add at most $3 \cdot |V|$ colours for those and one colour for each edge so we need at most $4 \cdot |V| + |E|$ colours for the entire reduction.
        With some modifications in the gadgets, we can strengthen this result by reducing the number of colours for jellies other than black jellies to only one. 
       
    \subsection{JELLY with one colour and black jellies}

        In this section, we slightly modify the previous proof so as to prove the following result:

        \begin{theorem}
            \label{JellyPHard2}
            1-Colour \textsc{Jelly} is PSPACE-hard when we allow the use of black jellies.
        \end{theorem}

        This result highlights the key role of black jellies in the complexity of 1-Colour \textsc{Jelly} since, in contrast, this very game without black jellies is known to be in NP \cite{JellyNo}.

        \begin{proof}
        We use the same reduction as for the proof of Theorem \ref{JellyPHard} (unbounded number of colours) but reduce the number of colours used to one. 

        Let $e = (u, v)$ be the edge for which we want to change the orientation, say from $(u, v)$ to $(v, u)$: that is to say, we want the level to be solvable even when the jelly representing $e$ is in the vertex gadget representing $u$ and cannot leave to go back to $v$. We give the same colour, say pink, to all the vertex jellies as well as to the jelly corresponding to $e$. We replace all the other jellies with black jellies.

        We start by computing a visibility representation of $G$ such that $u$ is the leftmost vertex in the representation, which is possible for any vertex \cite{Tamassia1986}.
        We then add a solving zone at the bottom of the level and ensure all vertex jellies, as well as the jelly representing $e$, can access this solving zone and merge. This requires a few modifications in our gadgets, which are as follows: 
        \begin{itemize}
        \item first, we add a hole of length 1 at the bottom of each vertex, leading to the solving zone at the bottom of the level. We shall call this hole a ``pit'': once the jelly enters the pit, it arrives at the solving zone and can't go back.
         \item we multiply the width of the edge jellies by two without changing their height; this certifies that these jellies will not fall in the pit, which has only length 1, while using the tunnels. 
         \item we remove the vertex jelly fixed inside the walls and replace the other vertex jelly by one of height 2 and width 1. Thus, this jelly can leave the vertex gadget once it arrives at the pit, but it cannot use the tunnels of height 1 which are meant to be used only by the edge jellies. 
    	\item last, we make an important change at the bottom of the vertex: completing the solving gap now requires two jellies of width 2 (representing red edges) or one of width 4 (representing blue edges).  Simply modifying the width of the original solving gap for vertex or blue bend gadgets would allow using a single jelly of width 2 as a moving platform across the gap, breaking the minimum inflow requirement. The new shape of the gap prevents this.
    	\end{itemize}
        Figures \ref{fig:new_and} and \ref{fig:new_red_bend} show the modifications for a vertex gadget and for a red bend gadget.
        	
        \begin{figure}[H]
            \centering

            \begin{subfigure}[t]{0.63\textwidth}
                \centering
                \includegraphics[width=0.89\textwidth]{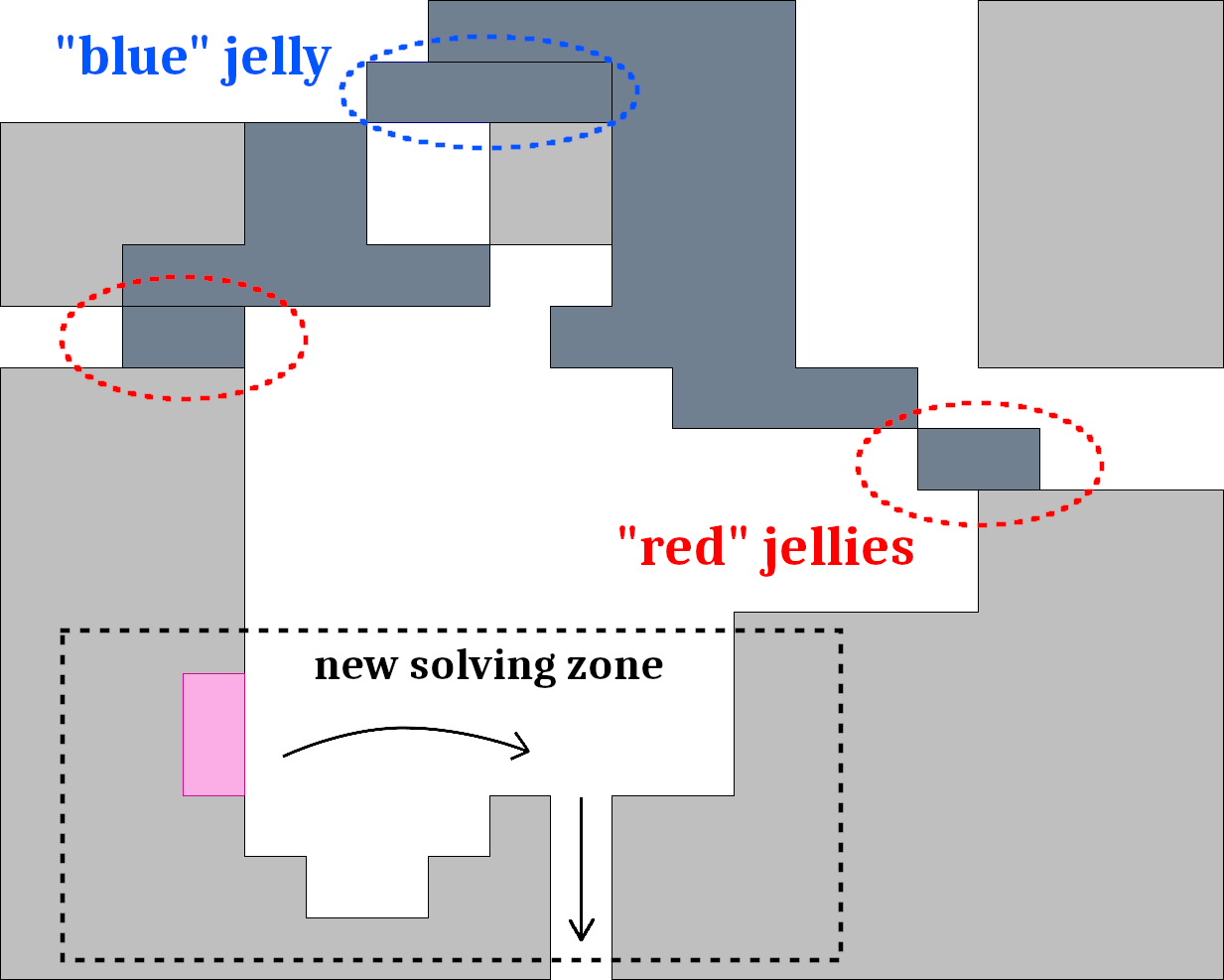}
                \subcaption[width=\textwidth]{New ($\overleftarrow{B}, \overleftarrow{R}, \overrightarrow{R}$) gadget}
                \label{fig:new_and}
            \end{subfigure}
            \begin{subfigure}[t]{0.01\textwidth}
                \;\;
            \end{subfigure}
            \begin{subfigure}[t]{0.34\textwidth}
                \centering
                \includegraphics[width=0.75\textwidth]{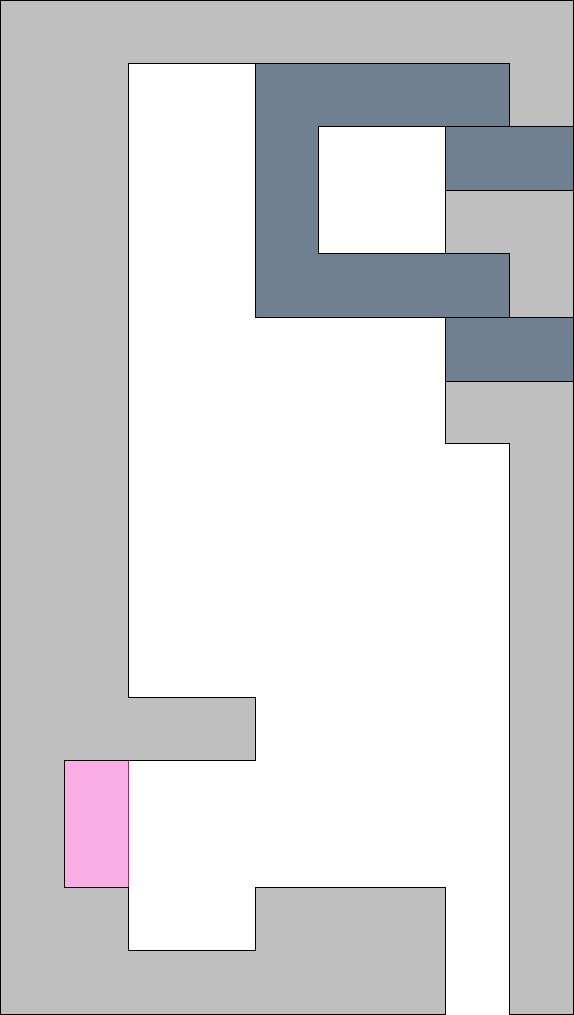}
                \subcaption[width=\textwidth]{New red bend gadget}
                \label{fig:new_red_bend}
            \end{subfigure}
            \caption{Modified gadgets. A gap, which ``blue" and ``red" jellies cannot fall through, leads vertex jellies to a global solving zone at the bottom of the level.}
        \end{figure}

        Now, for the original problem to have a solution, edge $e$ must be flipped from $(u, v)$ to $(v, u)$. This corresponds in our reduction to forcing the corresponding edge jelly, which is now pink, into the vertex gadget corresponding to target vertex $u$. We modify the solving zone in target gadget $u$ to add a tunnel of width 2 or 4 (depending on the weight of $e$) leading to the global solving zone. This tunnel is of height 1, thus the vertex jelly cannot go through it, and since $u$ is the leftmost part of the level there is no risk of crossing a tunnel or another vertex. 
        Figures \ref{fixed1} and \ref{fixed2} show this modification as well as solving mechanisms.
 
        \begin{figure}[H]
            \begin{subfigure}{0.49\textwidth}
                \centering
                \includegraphics[width=0.9\textwidth]{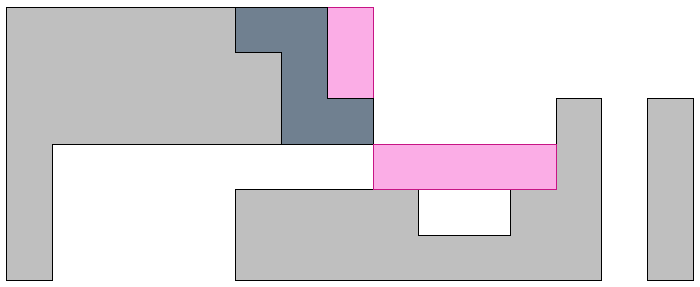}
                \subcaption[width=0.7\textwidth]{Letting the vertex jelly cross}
            \end{subfigure}
            \begin{subfigure}{0.49\textwidth}
                \centering
                \includegraphics[width=0.9\textwidth]{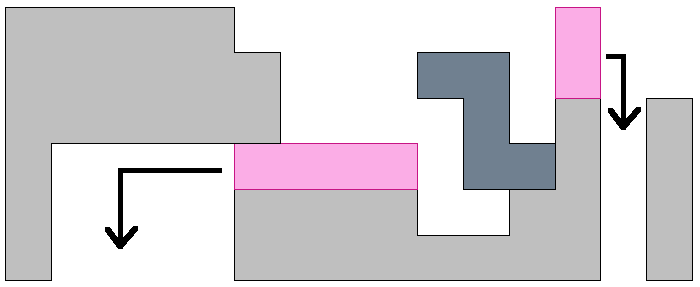}
                \subcaption[width=0.7\textwidth]{Jelly corresponding to $e$ can leave the vertex gadget}
            \end{subfigure}
            \caption{Modified solving zone in gadget for target vertex $u$ when $e$ is a blue edge}
        \label{fixed1}
        \end{figure}

        \begin{figure}[H]
            \begin{subfigure}{0.49\textwidth}
                \centering
                \includegraphics[width=0.9\textwidth]{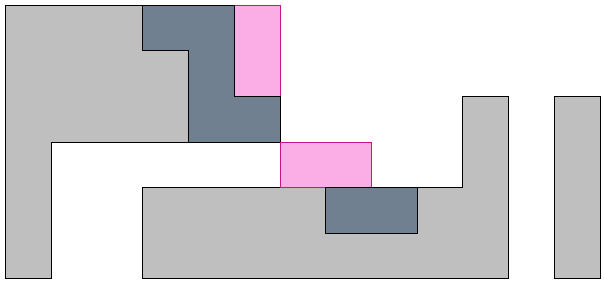}
                \subcaption[width=0.9\textwidth]{First solution}
            \end{subfigure}
            \begin{subfigure}{0.49\textwidth}
                \centering
                \includegraphics[width=0.9\textwidth]{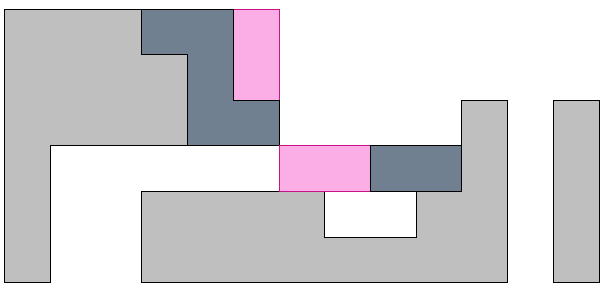}
                \subcaption[width=\textwidth]{Second solution}
            \end{subfigure}
            \caption{Modified solving zone in gadget for target vertex $u$ when $e$ is a red edge}
            \label{fixed2}
        \end{figure}
                       
        Similarly to our previous proof, we can easily see that the level constructed by our new reduction has a solution if, and only if, the original NCL problem has a solution.
        
        This concludes our proof of Theorem \ref{JellyPHard2}.
        \end{proof}

    \subsection{General JELLY is in PSPACE}\label{sec:inPspace}

    In this section, and for the sake of completeness, we prove the following result:
        
    \begin{theorem}
        \label{inPspace}
        General \textsc{Jelly} is in PSPACE.
    \end{theorem}

    \begin{proof}
        
    Let $J$ be an instance of \textsc{Jelly} of size $n$. Clearly, any configuration of $J$ can be stored using polynomial space. Let $G = (V, E)$ be the configuration graph of $J$: we represent all possible configurations of $J$ in a directed graph such that $(u, v)$ is in $E$ if and only if one can go from $u$ to $v$ in one step, i.e. by moving one block.
    
    Now, solving $J$ is equivalent to finding a path in $G$ between the initial configuration and a solution. There is no need to explicitely store the whole configuration graph as we can examine it step by step. This corresponds to the \textsc{ST-connectivity} problem, which is known to be in NL in relation to the size of the graph. 
    
    Since $G$ is exponential in $n$, \textsc{Jelly} is in NPSPACE. 
    Besides, by Savitch's theorem \cite{Savitch} NPSPACE equals PSPACE. Therefore \textsc{Jelly} is in PSPACE.
    
    This concludes our proof of Theorem \ref{inPspace}. 
    \end{proof}

\section{JELLY is NP-hard even with one bounded board dimension}\label{restrictedJelly}

       In this section, we prove that \textsc{Jelly} remains NP-hard even when we restrict one of the two dimensions of the board to be a small constant and even for jellies of one colour and no black jellies. We prove NP-hardness of 1-Colour \textsc{Jelly} with a bounded board height of nine in Section \ref{sec:boundedheight}, and with a bounded board length in Section \ref{par:jellyLength}.
       Membership in NP was proven for 1-Colour \textsc{Jelly}, even when both board dimensions are unbounded \cite{JellyNo}. Along with our new results, we obtain NP-completeness for both those problems.
	
	All of our reductions are from 3-\textsc{Partition}, which is the following problem:
	
		\textsc{Instance:} Let $m$ and $B$ be two integers. Let $S = \{a_1$,..., $a_{3\cdot{m}}\}$ be a set of $3\cdot{m}$ integers such that $\displaystyle\sum_{i=1}^{3\cdot{m}}a_i$ = $m\cdot B$, with $\forall i, \frac{B}{4} < a_i < \frac{B}{2}$. 
		
		\textsc{Question:} Can we divide $S$ in $m$ triplets $S_i = \{a_{i1}, a_{i2}, a_{i3}\}$ such that for all i, $a_{i1} + a_{i2} + a_{i3} = B$ ?  
		
		The basis of our reductions follows the original reduction in \cite{JellyNo}: there are $m$ horizontal or vertical gaps of length equal to the target $B$ and jellies whose length or height corresponds to the items of the 3-\textsc{Partition} instance; we need to fill the gaps with triplets of jellies merged together (horizontally or vertically).
		
		However, gravity in Jelly doesn't let us treat the case of bounded height and bounded length similarly: moves in the vertical direction are irreversible, whereas they are reversible in the horizontal direction. 
	Therefore, key points in the reductions differ: in the case of bounded length and unbounded height, we can vertically stack all the jellies one on top of the other (each one on its own platform) and select the order in which we drop them. Besides, any jelly standing alone on a platform can be dropped either from the left side or from the right. This lets us form (vertically) any triplet we want and also take the newly-formed jelly to any depth we like.
	In the case of bounded height however, where we stack jellies horizontally one next to the other, because of gravity, only the left-most jelly or the right-most one can fall (we can no longer select the order to drop the jellies and form the triplets). To overcome this difficulty, the original proof of \cite{JellyNo} used an unbounded height that depended linearly on $m$ (the number of gaps to be filled) and fixed platforms that lead to each gap. Thus, one could direct the next jelly in line in the correct gap following the correct vertical path of platforms. With a slightly cleverer arrangement of the platforms, one could manage to achieve the same thing in boards of height $\log m$. However, in order to achieve hardness for the case of constant height, this idea alone does not suffice. 
	
	In this case, our proof is based on the following key observation: when a jelly comes across a gap which is at least as large, it can either fall in or merge with some other jelly and pass over the gap. This property allows us to select the jellies to put in each gap (and thus in each triplet): we use holes of increasing length and a number of ``padding'' jellies which, when merged together, allow the newly-formed jelly to pass over some of the first smaller holes (the larger the jelly the further it can travel horizontally). The formation of triplets is thus possible by merging the appropriate number of paddings to a jelly so that it can pass over all the first holes and arrive to the one that corresponds to the selection. 

	In Section \ref{sec:2col}, we additionally show that 2-Colour \textsc{Jelly} is NP-hard for boards of a constant height of only 4, by using the property of jellies of different colour to move one on top of the other. However, let us note here that, unlike 1-Colour \textsc{Jelly}, it is not known whether this problem is in NP.

    \subsection{1-Colour JELLY is NP-complete even with a restricted number of lines}\label{sec:boundedheight}
    
    	We prove the following result:

        \begin{theorem}
            \label{linesJell}
            \textsc{Jelly} is NP-hard even with only one colour and nine lines.
        \end{theorem}

        \begin{proof}

		We reduce \textsc{3-Partition} to 1-Colour \textsc{Jelly}. 
		
		Let all the blocks of jelly be pink.
		In order to reduce \textsc{3-Partition} to 1-Colour \textsc{Jelly}, we must find a way to represent the integers and the triplets to be formed, with the idea that solving the part of the puzzle corresponding to some triplet should be equivalent to constructing a triplet in the original instance.

        For the following construction, we refer the reader to Figure \ref{1colbh}, which illustrates the reduction on a simplified instance with only two triplets.

        \begin{figure}[H]
            \centering
            \includegraphics[width = \textwidth]{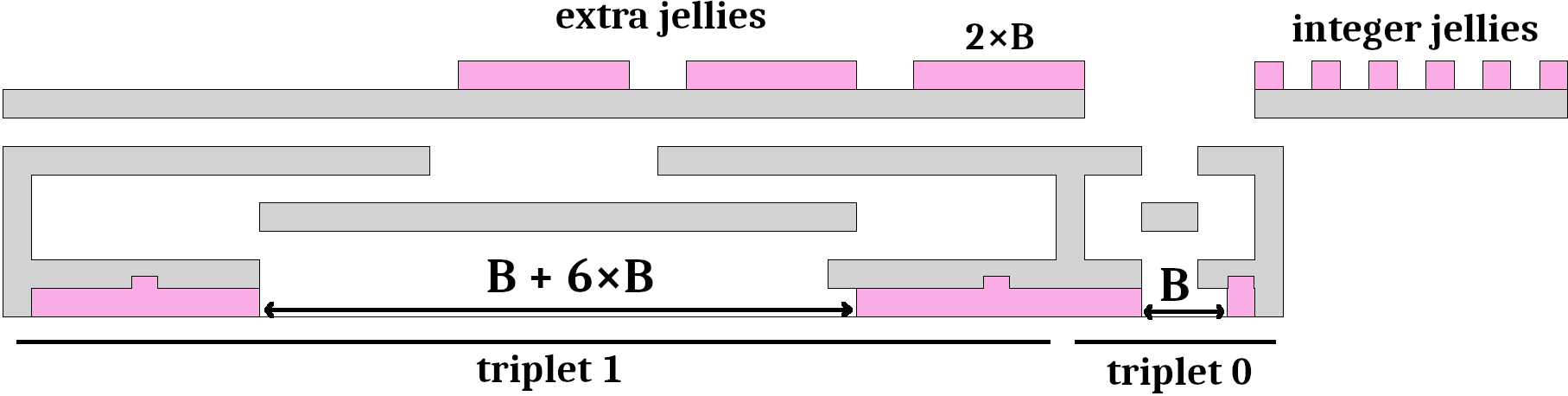}
            \caption{Example with $m = 2$ triplets and $B = 3$.}
            \label{1colbh}
        \end{figure}
        
        We represent integers $a_i$ as horizontal jellies of height 1 and width $a_i$ stored at the top right of the level on some continuous platform and order the triplets to be formed from $0$ to $m-1$. For each triplet $j > 0$, on the top left of the level, we add $3 \cdot j$ jellies of height 1 and width $2 \cdot B$. The idea of the reduction is that in order to access triplet $j$, we would need to use $3\cdot j$ of these additional jellies merged to the original jellies of length $a_i$. To achieve this goal, we introduce holes ever-increasing in length: to access the triplet $j$ we use a hole of length $2 \cdot j \cdot B + \frac{B}{2}$. Thus, a jelly of length $2\cdot j \cdot B + a_i$ cannot access the holes $0, \ldots, j-1$ but will fall through the $j^{th}$ hole.
        
Each triplet corresponds to two jellies stuck to the ground and separated by a horizontal gap. As with the lengths of the holes, these horizontal gaps are linearly increasing in length, starting from width $B$ and adding $6\cdot B$ to each gap from right to left: for triplet number $j$, this gap will be $(6 \cdot j + 1) \cdot B$. This triplet is accessible by a platform of same size $(6 \cdot j + 1) \cdot B$, allowing to choose whether the played jelly should fall from the left or the right. 

Since movable jellies falling to the ground directly attach to the jellies which are fixed on the ground, we introduce a gadget to make sure that the middle movable jelly can properly fill the gap, without getting stuck on top of the left or right movable ones. The left and right jellies stuck on the ground are covered by platforms able to support the largest movable jellies: the left one by a platform of width $2 \cdot j \cdot B + \frac{B}{2}$ and the right one by a platform of width $2 \cdot j \cdot B + \frac{B}{2} + 1$ to avoid falling jellies to merge immediately with it. The intention is that two jellies out of three should fall from the right, starting by the middle one and followed by the rightest one. The only jelly falling from the left side should be the last one used as it immediately merges and is therefore not allowed to move inside the gap. 

        Here, the hole for triplet 0 is of width $\frac{B}{2}$ and the one for triplet 1 of width $2 \cdot B + \frac{B}{2}$ and the gap for triplet 0 is $B$ whereas the one for triplet 1 is $6\cdot B +B$. For a level with three triplets, triplet 2 would have a gap of length $B + 12 \cdot B$ with a hole of width $4 \cdot B + \frac{B}{2}$, and on the top left we would need to add six extra jellies of length $2\cdot B$, two to allow each integer jelly to cross the hole of triplet 1.
        
        To place a jelly corresponding to some integer $a_i$, the player chooses which triplet $j$ it will be placed in and drops the corresponding extra jelly (except for triplet 0 whose width is exactly $B$), as visible on Figure \ref{bh1}. Then, the player horizontally merges the two jellies, allowing them to cross all the gaps between triplet 0 and triplet $j-1$. The jelly falls in the gap $j$, as visible on Figure \ref{bh2}, and can be placed accordingly.
        This mechanism allows the player to choose where to play each integer jelly and to complete the level, as visible on Figure \ref{bh3}, where each triplet contains three jellies from the original integer set.

        \begin{figure}[H]
            \centering
            \includegraphics[width = \textwidth]{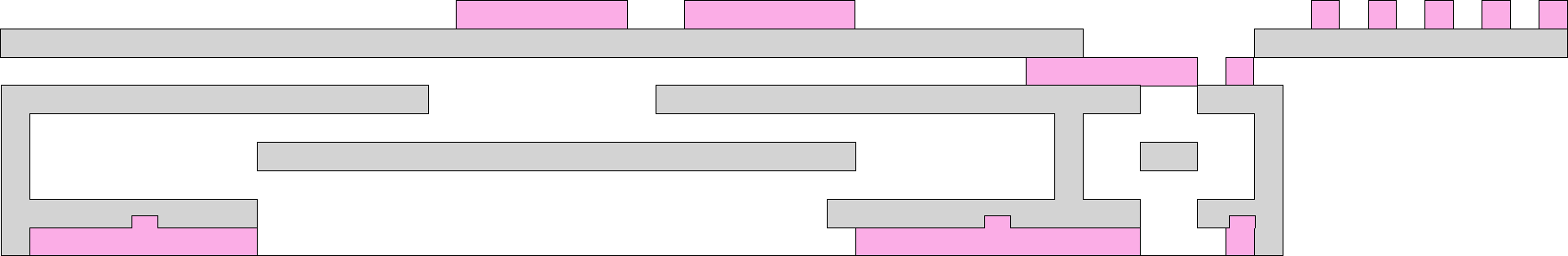}
            \caption{Choosing one extra jelly, corresponding to triplet 1.}
            \label{bh1}
        \end{figure}

        \begin{figure}[H]
            \centering
            \includegraphics[width = \textwidth]{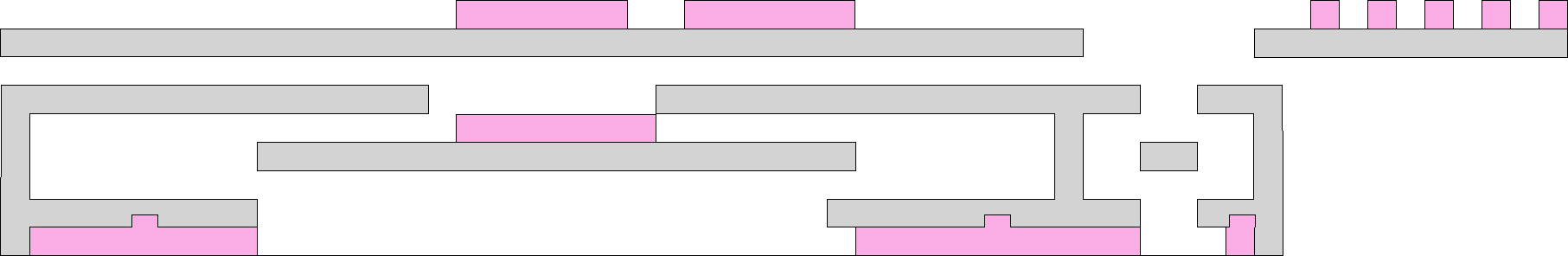}
            \caption{The new jelly can now pass the gap leading to triplet 0.}
            \label{bh2}
        \end{figure}

        \begin{figure}[H]
            \centering
            \includegraphics[width = \textwidth]{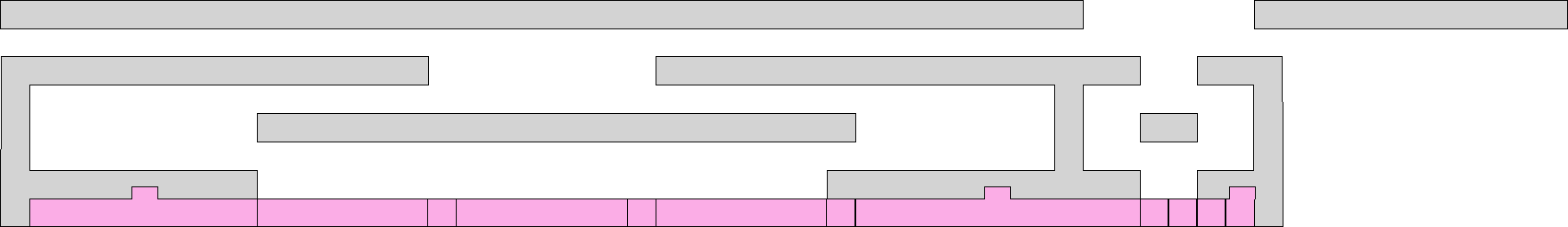}
            \caption{Final configuration, with each triplet gadget filled with exactly three integer jellies.}
            \label{bh3}
        \end{figure}
        
        It is easy to see that given a correct instance of \textsc{3-Partition}, the level constructed by our reduction is solvable, by filling each of the triplet gadgets as such: picking three jellies corresponding to the three integers of $S_i$, merging each one with the appropriate number of extra jellies to make them access the solving zone of the gadget and filling the gap.

		Now we must ensure that if the level of \textsc{Jelly} constructed by our reduction is solvable, then the corresponding instance of \textsc{3-Partition} has a solution. 
		Given any instance of \textsc{3-Partition}, we know that in the corresponding level created by our reduction and for every integer jelly $A_i$, $\frac{B}{4} < width(A_i) < \frac{B}{2}$. Besides, the total width of the integer jellies is exactly $m \cdot B$. This implies that the only way to reach a total length of $B$ is to sum the lengths of exactly 3 jellies.
		
		Since the width of these extra jellies is an even multiple of B, the triplet gadgets cannot be filled with these alone. Using exactly three integer jellies is required to form the remaining width $B$. 
		Now, the only way to solve the level is to connect all the pink jellies by filling each of the $m$ gadgets with three jellies whose total width is exactly $B$ and some extra jellies. These jellies must all be of height exactly 1: if a jelly of height bigger than 1 is constructed, this means that two of the movable jellies merged vertically instead of horizontally and either the required width cannot be reached or another jelly needs to be used in the gadget, meaning that there will be a jelly missing somewhere else in the level. In total, the integer jellies are grouped in $m$ triplets each of equal sum $B$. This corresponds exactly to a solution of \textsc{3-Partition}. 

        This concludes our proof of Theorem \ref{linesJell}.
        \end{proof}

        Note that the reduction is quadratic: for each triplet, we must widen the gap by $6 \cdot B$ blocks and add the two spaces on each side of the gap. In total, the width of the level is $\mathcal{O}(B \cdot m^2)$. This remains a polynomial reduction, proving Theorem \ref{linesJell}.
        However, the question remains open for levels of height smaller than nine. 

   \subsection{2-Colour JELLY is NP-hard even with a restricted number of lines}\label{sec:2col}

        In this section, we show NP-hardness of \textsc{Jelly} with boards of only three lines by allowing the use of a second colour. This also allows for a linear reduction, instead of quadratic for our previous result.
        
        Let us now consider a variant of \textsc{Jelly} where jellies can be of two possible colours (or one colour and black jellies), and the height of the board is limited to three lines. We prove the following result:

        \begin{theorem}
            \label{thm:2col}
            \textsc{Jelly} is NP-hard even with only two colours and three lines. Besides, we only need one jelly from the second colour.
        \end{theorem}

        \begin{proof}
        Once again, we start from an instance of \textsc{3-Partition} and our construction is very similar to the one in \cite{JellyNo}. However, adding one colour allows to easily restrict the height of the board, whereas the original construction required a height of $O(\log m)$ with a necessity to add horizontal platforms acting as an equivalent to our choice zone. Given two colours, say pink and blue, we start by representing the integers $a_i$ as pink jellies of height 1 and of width $a_i$. We store those jellies on a continuous platform at the very left of the level. A triplet gadget is composed of two pink jellies of size $1 \times 2$, stuck so that they cannot move and separated by a gap of width $B$, assuming $B$ is strictly greater than 0 (otherwise the problem is trivial). The bottom part of the level is composed of $m$ such gadgets. On top of those, a blue jelly of height 1 and width $B+1$ acts as a movable platform to allow the jellies representing the integers to fill the triplet gadgets. See Figure \ref{twocoljel}.

        \begin{figure}[H]
            \centering
            \includegraphics[width=0.9\textwidth]{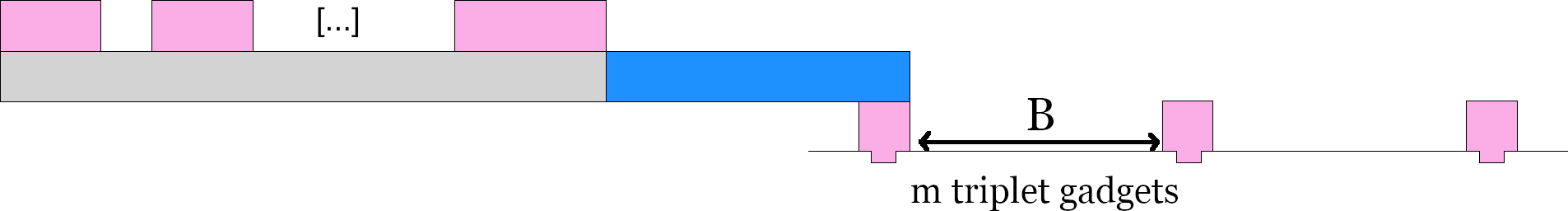}
            \caption{Adapted reduction for three lines and two colours}
            \label{twocoljel}
        \end{figure}

        \vspace{6pt}

        Once again, all the pink jellies can be merged if and only if the jellies at the top of the level represent integers for which \textsc{3-Partition} has a solution. 
        
        This concludes our proof of \ref{thm:2col} and our work on \textsc{Jelly} with a bounded height. 
        \end{proof}
        
        \subsection{1-Colour JELLY is NP-hard even with a restricted number of columns}\label{par:jellyLength}

		We now consider 1-Colour \textsc{Jelly} when the length of the board is restricted to five columns. We prove that the game is still NP-complete. Observe that the constant is much smaller than the one given in Theorem \ref{linesJell}. Furthermore the reduction is linear instead of quadratic. This highlights the asymmetry between horizontal and vertical restrictions mentioned in the introduction of Section \ref{restrictedJelly},

		Once again, membership in NP is known \cite{JellyNo}. To obtain NP-completeness, we show the following result:

		\begin{theorem}
		    \label{columns}
		    \textsc{Jelly} is NP-hard even with only one colour and five columns.
		\end{theorem}

		\begin{proof}
		We proceed by reduction from the \textsc{3-Partition} problem.
		Let all blocks of jelly be pink.
		We represent each $a_i \in S$ by a jelly $A_i$ of width 1 and height $a_i$. At the beginning of the game, we store these jellies at the top right of the level, on top of each other separated by a platform of width 1. This way, the player can move the jellies $A_i$ in any order, and the first move of a jelly is always to fall on the left.

		We represent triplets by triplet gadgets of width 5 and height $B+2$ (see Figure \ref{reduction1} and Figure \ref{reduction2}).     
		Two jellies of size $2\times{1}$ are separated vertically by a space of height $B$. Grey walls ensure that they cannot move and that they do not fall. On the right of the bottom block, a platform allows the movable jellies to fall next to the space to be filled. To solve the puzzle, the player must fill this space with movable jellies in order to merge the two jellies into one. At the bottom right corner of the triplet, a hole leads to the triplet below, except for the one at the bottom of the board. Due to gravity and the disposition of the walls, the only way to connect the two jellies is to fill the space between them with a jelly of height exactly $B$. 
		        
		Two triplets are linked by a choice zone of height $\frac{B}{2}$ (or the height of the tallest movable jelly available at the beginning). It is composed of a jelly of height $\frac{B}{2}$ and width 1 at the very left that connects the bottom jelly of the upper triplet with the top jelly of the lower triplet, and is not accessible to the player due to the hole at the top of the lower triplet.

		    \begin{figure}[H]
		        \centering
		         \begin{subfigure}[t]{0.39\textwidth}
		            \centering
		             \includegraphics[height=0.35\textheight]{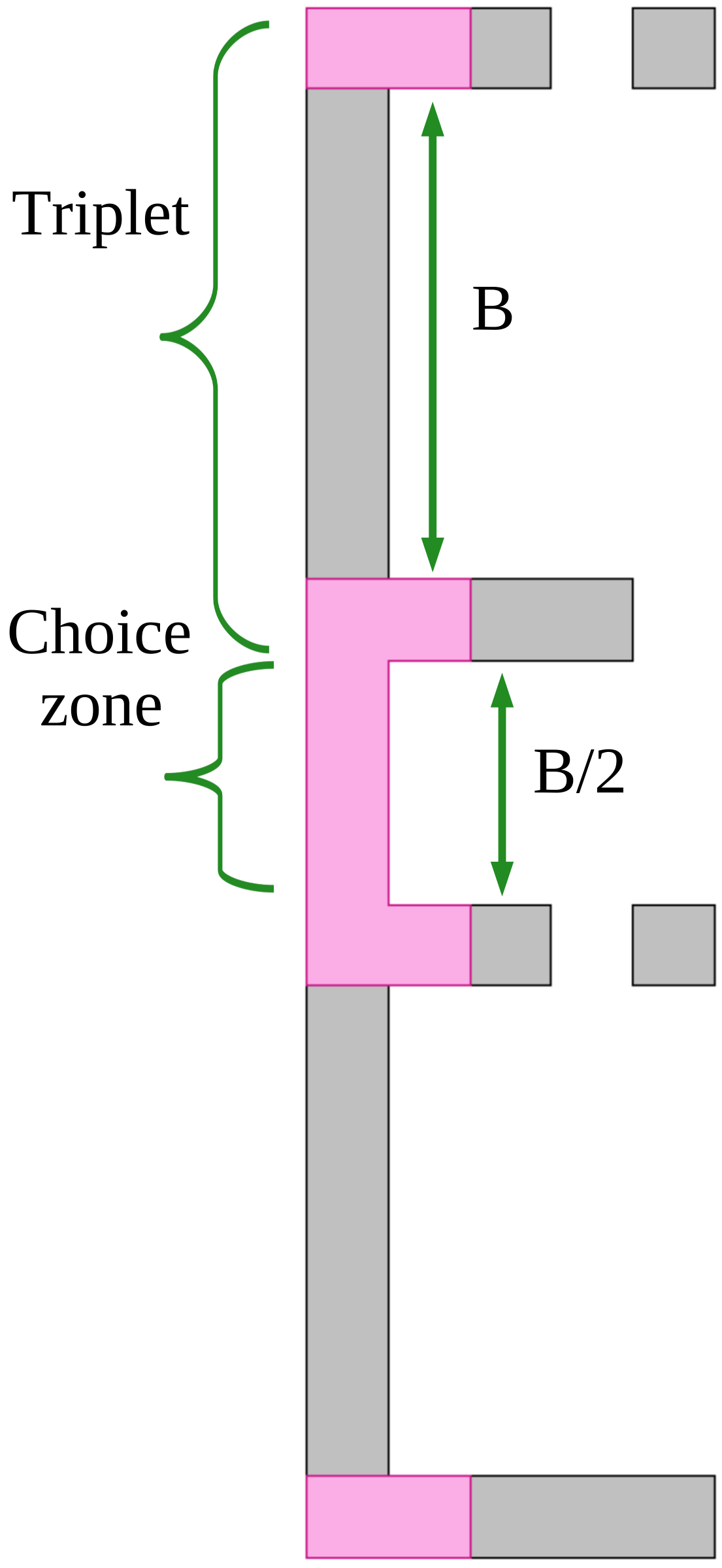}
		             \subcaption{\centering Two triplet gadgets linked by a choice zone}
		             \label{reduction1}
		         \end{subfigure}
		         \begin{subfigure}[t]{0.02\textwidth}
		            \centering
		            \;\;\;\;\;
		         \end{subfigure}
		         \begin{subfigure}[t]{0.57\textwidth}
		            \centering
		             \includegraphics[height=0.37\textheight]{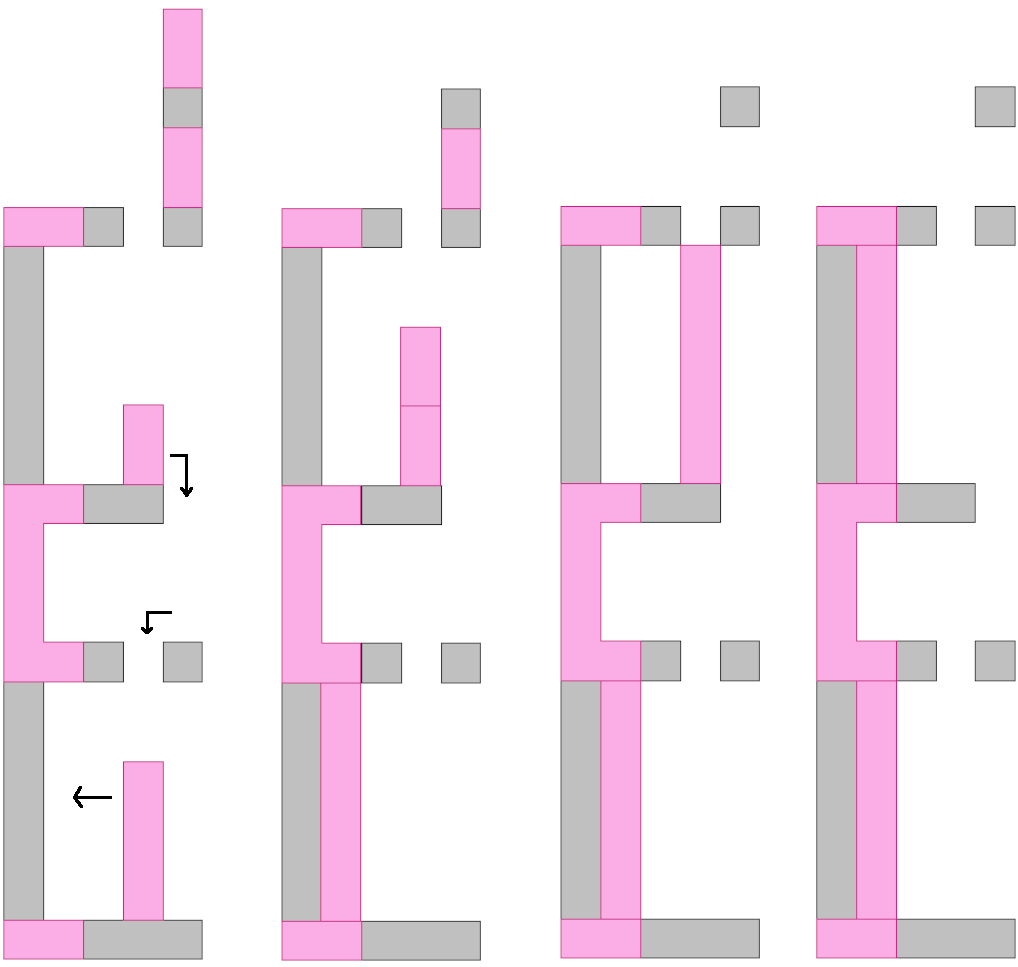}
		             \subcaption{Filling triplet gadgets using integer jellies stored at the top right of the board.}
		            \label{reduction2}
		         \end{subfigure}
		         \caption{Reduction to Jelly-No}
		    \end{figure}

		This choice zone allows to keep the width of the board to a constant of 5 columns while still making it possible for the player to move the current jelly any way needed.  For an instance of size $3\cdot m$, we need $m$ triplets.
		
		A key difference with the original construction from \cite{JellyNo}, which required a level of dimensions $O(B \cdot m) \cdot O(\log m)$, is the way gravity affects the solving process. Since moving a block horizontally (to the right or left) on some continuous platform is reversible, the choice zone between two gadgets allows to restrict our level to a constant number of columns.

		It is easy to see that solving the constructed level is possible if and only if the original instance has a solution.
		
		This concludes our proof of Theorem \ref{columns}.
		\end{proof}

\section{HANANO is NP-hard even with one bounded board dimension}\label{restrictedHanano}
        
    In this section we prove that 1-Colour \textsc{Hanano} remains NP-hard even with one bounded board dimension (either length or height). For these problems though, let us note that, unlike 1-Colour \textsc{Jelly}, it is not known whether they are in NP. Besides, we make use of movable grey blocks, which act similarly to black jellies from \textsc{Jelly}. 
    
In Section \ref{sec:hananoLines}, we prove that 1-Colour \textsc{Hanano} with bounded board height is NP-hard. In Section \ref{par:hananoLength} we adapt our proof of Section \ref{par:jellyLength} to show that 1-Colour \textsc{Hanano} with bounded board length is NP-hard. 
        
    \subsection{1-Colour HANANO is NP-hard even with a restricted number of lines}\label{sec:hananoLines}
    
        We prove the following result:
    	
    	\begin{theorem}
    		\label{hananoLines}
    		\textsc{Hanano} is NP-hard even with only one colour and eleven lines. Furthermore, we only need one coloured block and one flower block.
    	\end{theorem}
    	
    	\begin{proof}
		We consider the Numerical 3-Dimensional Matching, also called the \textsc{ABC-Partition} problem, an NP-hard variant of \textsc{3-Partition} \cite{gareyjohnson}. The problem is the following:
		
		\textsc{Instance:} Two integers $m$ and $B$. Three disjoint multisets $X, Y, Z$ of $m$ integers each, such that $\displaystyle\sum_{a_i \in X\cup Y\cup Z}a_i$ = $m\cdot B$ and $\forall a_i \in X\cup Y\cup Z$, $\frac{B}{4} < a_i < \frac{B}{2}$.
		
		\textsc{Question:} Can we divide $(X\times Y \times Z)$ in $m$ triplets $S_i = \{x, y, z\}$ where $i = [1, \ldots, m], x \in X$, $y \in Y$ and $z \in Z$, such that for all $i = [1,\ldots, m], x + y + z = B$?
		
		Once again, we represent triplets as gaps to be filled with the appropriate movable grey blocks in order to allow a coloured block to cross the gap and reach a flower. We must ensure not only that each triplet contains one, and only one integer from each set, but also that the sum of each triplet is strictly equal to $B$. This requires to force the integer blocks in certain positions to prevent them from sliding horizontally, which would allow the red block to cross a gap without completely filling it. To enforce these conditions, we form our reduction as follows.
		
		A triplet gadget is a gap of length $B$ and height $3$, separating two tunnels leading to the previous (left) and next (right) triplets. The first triplet, on the left, contains a red block; the last one, on the right, contains a flower stuck to the ground which cannot move. 
		At the very left and very right of each triplet gadget, there is a gap of height and length 1. In the middle of the gadget, there is an unmovable platform of size $(B-2)\times 1$; on top of it there is a movable block of size $1 \times 1$. See Figure \ref{hananoLGadg} below.
		
		\begin{figure}[H]
		    \centering
		    \includegraphics[width=\textwidth]{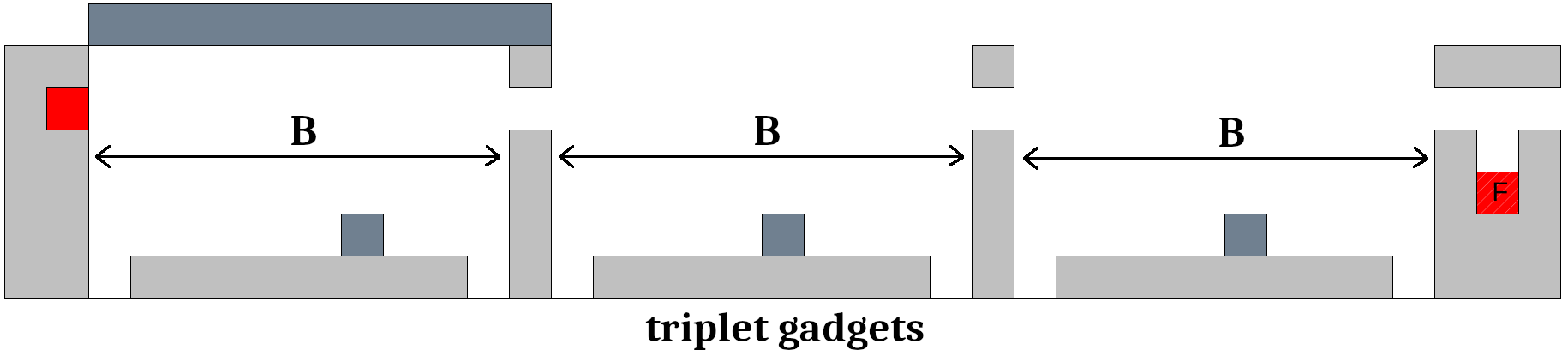}
		    \caption[]{3 triplet gadgets, with the movable platform to travel through them.}
		    \label{hananoLGadg}
		\end{figure}
		
		We represent the integers using movable blocks (represented in dark grey), of length equal to the integer they represent, and store them on a fixed platform at the very left of the level, with a movable platform allowing them to travel through the triplet gadgets. The shapes of the integer blocks vary depending on which set $X, Y$ or $Z$ the integers belong to. For integers $x_i \in X$ and $z_i \in Z$, the corresponding blocks will be of height $4$ and have the shape of a $\Gamma$ and a horizontally reversed $\Gamma$ respectively. These two blocks can fall to the left and right of each gadget to form parts of the bridge. The blocks corresponding to integers $y_i \in Y$ will have the same shape as the ones in $X$ but only be of height 2. 
				
		Clearly, if there exists a solution to the original instance of \textsc{ABC-Partition}, then the corresponding level can be completed according to this solution, by filling each gap with one block of each type: one block representing some $x_i \in X$ and one representing some $z_i \in Z$ at the left and right of the gadget, and using the extra block of size $1 \times 1$ to allow a block representing some $y_i \in Y$ to reach the right height to complete the bridge. See Figure \ref{hananoLFilled}.
		
		\begin{figure}[H]
		    \centering
		    \includegraphics[width=\textwidth]{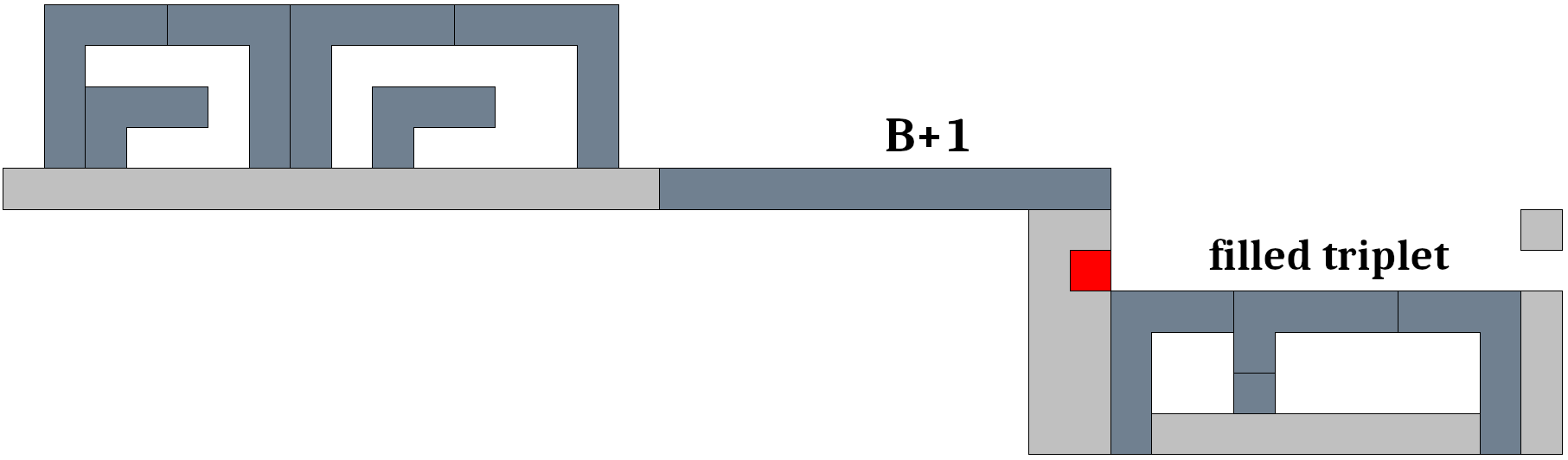}
		    \caption[]{Integer blocks stored at the top left of the level and triplet filled with one block of each shape, for sets $X$, $Y$ and $Z$ from left to right.}
		    \label{hananoLFilled}
		\end{figure}
		
		It remains to show that if the corresponding level can be completed, then the original instance has a solution. This requires to prove two things: first that it is necessary to use one block of each type in each triplet gadget, and second that the total length of these blocks must be exactly $B$.
		
		First, as visible on Figure \ref{fig:hanano_XYZ_impossible}, it is impossible to solve a triplet gadget using two $X$ integer blocks, as there is only one gap on the left to bring such block to the right height, to let the coloured block cross. This reasoning holds for $Z$. 

		\begin{figure}[H]
			\centering
			\includegraphics[width=0.4\textwidth]{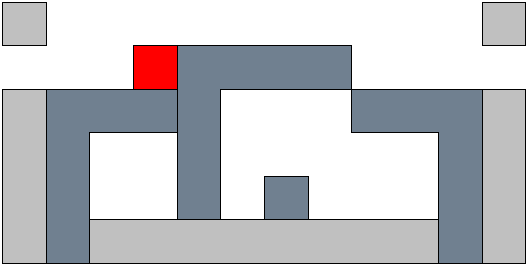}
			\caption{Trying to use two $X$ integer blocks for a same triplet}
			\label{fig:hanano_XYZ_impossible}
		\end{figure}
		
		Therefore, each triplet can have at most one $X$ block and one $Z$ block. Besides, once those blocks are in place, they cannot move.
		Now, due to the constraint on integers $a_i$, two blocks are not enough to fill a gap of length $B$: thus we need at least one $Y$ block to complete the gap.
	    Since we have $m$ blocks of each type for $m$ triplets, this means we must use exactly one block of each type in each triplet.
		
		Once a triplet is filled with one block of each type, none of these blocks can move without preventing the red block from crossing the gap, even if their total width is strictly smaller than $B$: the blocks of type $X$ and $Z$ must be fixed in the ground and once the block of type $Y$ is on the extra $1 \times 1$ block, it cannot move without falling. Thus using one block of each type requires their total width to be exactly $B$.

		Therefore solving our level is possible if, and only if, there exists a solution to the original instance.
		Besides, once again, the reduction is linear.
		
		This concludes our proof of Theorem \ref{hananoLines}.
    	\end{proof}
        

	\subsection{1-Colour HANANO is NP-hard even with a restricted number of columns}\label{par:hananoLength}

		In this section, we adapt our reduction of Section \ref{par:jellyLength} to prove that 1-Colour \textsc{Hanano} is also NP-hard with a bounded number of columns. 
		
		We consider a variant of \textsc{Hanano} where all the blocks and flowers have the same colour, the blocks can only bloom upward and the length of the board is restricted to six columns. We prove the following result:

		\begin{theorem}
		    \label{columnsHan}
		    \textsc{Hanano} is NP-hard even with only one colour and six columns.
		\end{theorem}

		\begin{proof}
		We show once again that our problem is NP-hard by reduction from the \textsc{3-Partition} problem. 
		
		We want to construct a level of \textsc{Hanano} such that the level is solvable if and only if the corresponding instance of \textsc{3-Partition} is solvable. 

		Once again, we assume that all the flowers and blocks to be bloomed are red, walls or platforms, represented in light grey, are unmovable and dark grey blocks can be moved by the player, are affected by gravity and do not need to bloom in order for the level to be completed.
		We represent each $a_i \in S$ by such a movable grey block $A_i$ of width 1 and height $a_i$ and we store these blocks at the top right of the level, on top of each other. Again, we separate them by a platform of width 1 but this time this choice is just for readability purposes. 

		We adapt the triplet gadget and choice zone as to fit \textsc{Hanano} (see Figures \ref{hanano1} and \ref{hanano2}). Each gadget contains one red block on the very left and one red flower separated by a gap of size $B \times 1$. Triplets are separated by a choice zone that allows grey blocks to either go to the current gadget to fill the gap or continue to the ones below. 

		\begin{figure}[h]
		     \begin{minipage}[t]{0.50\textwidth}
		         \centering
		         \includegraphics[width=0.8\textwidth]{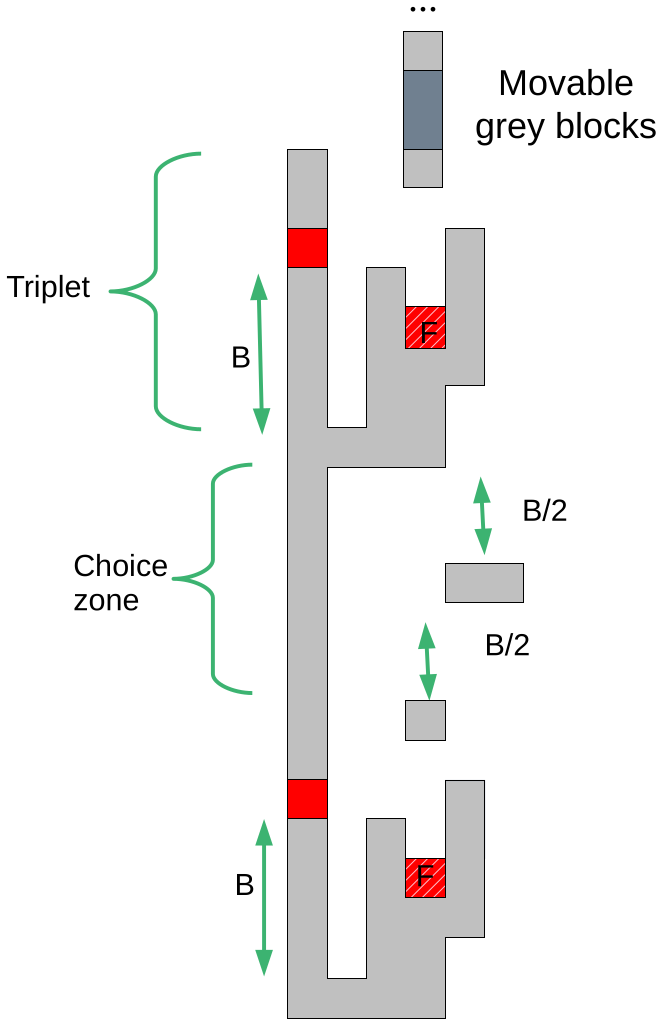}
		         \caption[width=\textwidth]{Triplet gadget and choice zone for Hanano, movable grey blocks $A_i$ at the top}
		         \label{hanano1}
		     \end{minipage}
		     \begin{minipage}[t]{0.50\textwidth}
		         \centering
		         \includegraphics[width=0.8\textwidth]{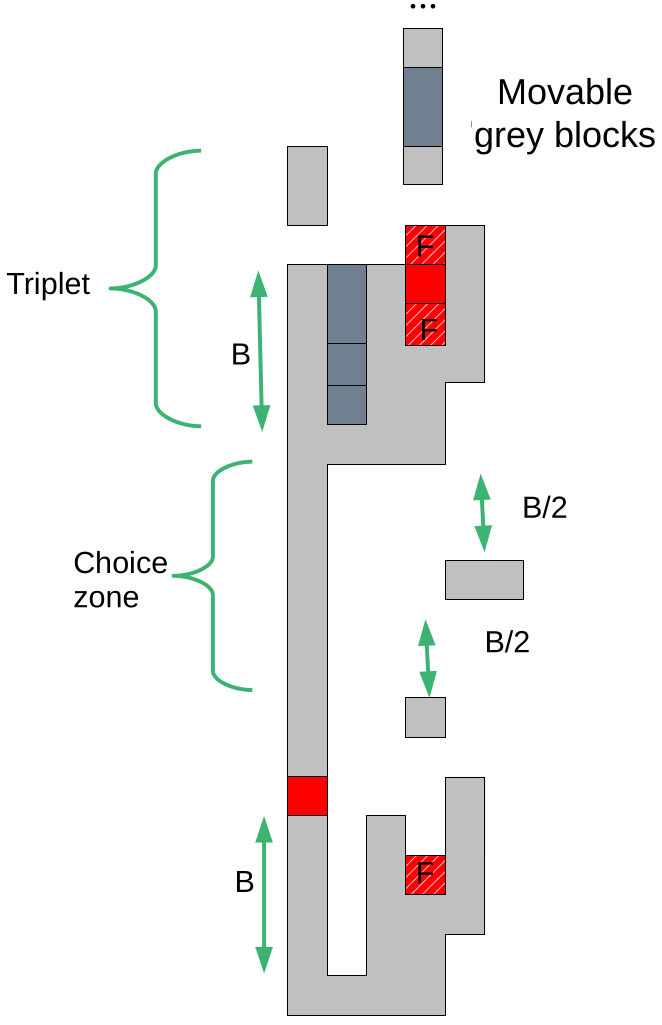}
		         \caption[width=0.9\textwidth]{Completed top triplet gadget.}
		        \label{hanano2}
		     \end{minipage}
		\end{figure}
		 
		 A level is once again composed of $m$ triplet gadgets separated by choice zones, and grey blocks stocked at the top. In order for the block on the left to reach the corresponding flower and bloom, the gap must be filled with blocks of a total height of $B$. 
		 At the beginning of the game, the only way for a red block to move is to go to the right. If the gap there is not filled then the block falls and can never bloom, so filling the space with a total height of $B$ is the only way to make it bloom and ensure the level can be completed. Moreover, the walls placed around the flower ensure that once a red block has bloomed it finds itself stuck. It cannot move to leave the triplet and possibly interact with other red blocks, so the only way for a given red block to bloom is to touch the flower corresponding to its triplet. 

		The rest of the proof is similar to the one in Section \ref{par:jellyLength}. If the instance of \textsc{3-Partition} has a solution then in the corresponding level, one can divide the grey blocks corresponding to $a_i$ to fit the $m$ gaps of height $B$, therefore reaching a total height of $m\cdot{B}$. 

		Now we must ensure that if the level of \textsc{Hanano} constructed by our reduction is solvable, then the corresponding instance of \textsc{3-Partition} has a solution. 
		Given any instance of \textsc{3-Partition}, we know that in the corresponding level of \textsc{Hanano} the only way to reach a total height of $B$ is to sum the heights of exactly 3 blocks. 
		Now, the only way to solve the corresponding level of 6 columns \textsc{Hanano} is to fill each of the $m$ gadgets with three grey blocks of total height exactly $B$, so a total height of $m\cdot B$: indeed, if the gap were filled with blocks of a height of more than $B$, it would also prevent the red block from reaching the flower.
		Since this is exactly the sum of the heights of all the grey blocks, the only way to solve the level is to fill each gadget with exactly three blocks. Besides, due to gravity, once a grey block is used in one gadget it cannot leave to go to another gadget. In total, we construct $m$ triplets of equal sum of heights $B$. This corresponds exactly to a solution of \textsc{3-Partition}. 
		
		This concludes our proof of Theorem \ref{columnsHan}.
		\end{proof}
	
\section{Conclusion and open questions}

    Our main result is that allowing black jellies renders \textsc{Jelly} PSPACE-complete even when all the other jellies have the same colour, by indirect reduction from the \textsc{NCL} problem, using related work on \textsc{Hanano} and the concept of visibility representations to circumvent the effects of gravity.
    
    We have also shown that both \textsc{Jelly} and \textsc{Hanano} are NP-hard even with one colour and a restricted number of columns, as well as for a restricted number of lines. In the case of \textsc{Jelly}, allowing a second colour results in a stronger bound and going from a quadratic to a linear reduction.
    
    However, the following open questions remain:
    \begin{enumerate}
        \item {Is there a constant $c$ such that 1-Colour \textsc{Hanano} with $c$ columns, or with $c$ lines, is in NP? For which we can prove PSPACE-hardness?}
        \item {Is \textsc{Jelly} with a constant, greater than one, number of colours still PSPACE-complete if black jellies cannot be used?}
        \item {Is 1-Colour \textsc{Jelly} with black jellies still PSPACE-hard if black jellies are rectangular? Same question for the NP-hardness of 1-Colour \textsc{Hanano} with a restricted number of lines.}
        \item {Is there a constant $c$ such that restricting 1-Colour \textsc{Jelly} or 1-Colour \textsc{Hanano} to levels of $c$ columns or $c$ lines allows for a polynomial-time algorithm?}
    \end{enumerate}
    
    The mechanism of gravity also opens interesting questions about the complexity of the games when this constraint is ignored or at least relaxed. The puzzle game Yugo, a variant of Jelly-No where jellies are allowed to move vertically by climbing on half-blocks, was recently shown by Naoya Kano, Katsuhisa Yamanaka and Takashi Hirayama \cite{yugo} to be PSPACE-complete with an unbounded number of colours. It would be interesting to know whether the 1-Colour restriction of this game remains PSPACE-complete.
    
    \medskip
    
\thanks{\textbf{Acknowledgements:} We would like to thank the anonymous reviewers for their careful reading of our manuscript. In particular, we would like to thank reviewer 1 for making a suggestion that significantly simplifies our proof of Theorem 2.1.}    
    

\bibliographystyle{unsrt}
\bibliography{refs} 


\end{document}